\documentclass[letterpaper,twocolumn,10pt]{article}
\usepackage{usenix2020}

\usepackage{amsmath,amsthm,amssymb,amsfonts,bm,bbm}
\usepackage{graphicx,subcaption}
\usepackage{enumerate,color,soul,makecell}
\usepackage{gensymb}
\usepackage{mathtools}
\usepackage{indentfirst}
\usepackage{algorithm}
\usepackage{algorithmic}

\usepackage{tikz}

\usepackage{multirow}
\newtheorem{theorem}{Theorem}
\newtheorem{definition}{Definition}
\newtheorem{lemma}{Lemma}

\usepackage{pifont}

\newcommand{\xmark}{\ding{55}}
\usepackage{threeparttable}

\author{{\rm Xiaolan Gu}\\ University of Arizona\\{\rm xiaolang@email.arizona.edu}
\and {\rm Ming Li}\\ University of Arizona\\{\rm lim@email.arizona.edu} 
\and {\rm Li Xiong}\\ Emory University\\{\rm lxiong@emory.edu}}

\begin{document}

\date{}
	
\title{\Large \bf PRECAD:  Privacy-Preserving and Robust Federated Learning via \\ Crypto-Aided Differential Privacy}

\maketitle

\begin{abstract}
Federated Learning (FL) allows multiple participating clients to train machine learning models collaboratively by keeping their datasets local and only exchanging model updates. Existing FL protocol designs have been shown to be vulnerable to attacks that aim to compromise data privacy and/or model robustness. Recently proposed defenses   focused on  ensuring either privacy or robustness, but not both. In this paper, we develop  a framework called PRECAD, which simultaneously achieves differential privacy (DP) and enhances robustness against model poisoning attacks with the help of cryptography.  Using secure multi-party computation (MPC) techniques (e.g., secret sharing), noise is added to the model updates by the   honest-but-curious server(s) (instead of each client) without revealing   clients' inputs,  which achieves the benefit of centralized DP in terms of providing a  better privacy-utility tradeoff than local DP based solutions. Meanwhile, a crypto-aided secure validation protocol is designed to verify that the contribution of model update from each client is bounded without leaking  privacy. We show analytically that the  noise added to ensure DP also provides enhanced robustness against malicious model submissions. We experimentally demonstrate that our PRECAD framework achieves higher privacy-utility tradeoff and enhances robustness for the trained models.

\end{abstract}

\section{Introduction}
\label{sec:introduction}

\begin{table*}[t!]
	\footnotesize
	\centering
	\caption{Comparison of differentially private FL models and approaches}
	\renewcommand{\arraystretch}{2.0}
	\vspace{-3mm}
	\begin{threeparttable}
	\begin{tabular}{c|cc|ccc|cc}
		\Xhline{1pt}
		\multirow{2}{6em}{\makecell[c]{Privacy Model\\ and Protocol(s)}} & \multicolumn{2}{c|}{Threat Model of Privacy} &
		 \multicolumn{3}{c|}{Privacy-Utility Tradeoff \tnote{$\mathsection$} ~~(with record-level DP)} & \multicolumn{2}{c}{Robustness against Poisoning Attacks}  \\
		 \cline{2-8}
		 & Server & \makecell[c]{Clients \\ Collusion} &
		\makecell[c]{Noise \\ Generator } &
		\makecell[c]{Perturbation \\ Mechanism} &
		 \makecell[c]{STD of Noise \\ in Aggregation} &  \makecell[c]{Mechanism\\   } & \makecell[c]{Impact of\\  Added Noise}  \\
		\Xhline{1pt}
		\makecell[c]{CDP\\ \cite{mcmahan2018learning,geyer2017differentially} }&  trusted &  all except victim  & server & $\sum_i Q_i+\mathcal{N}(0,\sigma^2)$ &   $\sigma$  & \makecell[c]{via robust\\ aggregation} & \makecell[c]{noise enhances\\ robustness} \\
		\hline
		\makecell[c]{LDP \\ \cite{li2020differentially,pihur2018differentially}} & untrusted & all except victim  & client & $Q_i+\mathcal{N}(0,\sigma^2)$ &   $\sqrt{n}\cdot\sigma$  & \makecell[c]{via robust\\ aggregation} & \makecell[c]{noise reduces\\ robustness} \\
		\hline
		\makecell[c]{DDP+Crypto\tnote{*}\\ \cite{truex2019hybrid,xu2019hybridalpha,lyu2020lightweight}} & honest-but-curious\tnote{$\diamond$} & $\tau$ non-collude & client & $Q_i+\mathcal{N}(0,\frac{\sigma^2}{\tau})$ & $\sqrt{\frac{n}{\tau}}\cdot\sigma$ &  \xmark & --\\
		\hline
		\makecell[c]{\textbf{PRECAD}\\  (ours)\tnote{$\dagger$}} & \makecell[c]{honest-but-curious\tnote{$\diamond$}\\ non-colluding} & all except victim & \makecell[c]{two \\ servers} & \makecell[c]{$\sum_i Q_i+\mathcal{N}(0,2\sigma^2)$ \\ or $\sum_i Q_i+\mathcal{N}(0,\sigma^2)$} & \makecell[c]{$\sqrt{2}\cdot\sigma$ \\ or $\sigma$} & \makecell[c]{via MPC\\ and DP} & \makecell[c]{noise enhances\\ robustness} \\
		\Xhline{1pt}
	\end{tabular}
	\begin{tablenotes}
     \item[*] DDP+Crypto framework assumes a minimum number of non-colluding clients (i.e., the value of $\tau$) out of $n$ total clients, which has influence on the privacy guarantees. Symbol ``--'' indicates not applicable.
     \item[$\diamond$]  Honest-but-curious means that the server follows protocol instructions honestly, but will try to learn additional information.
     \item[$\mathsection$] We show the privacy-utility tradeoff by fixing the same privacy cost (with required noise) and then comparing the standard deviation (STD) of the noise on the aggregation, where approaches with a smaller STD   has a better utility. $Q_i$ represents the local model update (without noise) from client $\mathsf{C}_i$. For convenience, we ignore the scaling factor for averaging local models, and let the record-level sensitivity of all $Q_i$ be 1, which can be achieved via record-level clipping. 
     \item[$\dagger$] PRECAD considers two types of attacker settings on privacy, i.e., with/without a corrupted server. The latter   assumption is weaker, thus needs less noise.
   \end{tablenotes}
   \end{threeparttable}
	\label{tab:DP_compare}
\end{table*}

Federated learning (FL) \cite{mcmahan2017communication}  is an emerging paradigm that enables multiple clients to collaboratively learn models without explicitly sharing their data. The clients upload their local model updates to the server, who then shares the global average with the clients in an iterative process. This offers a promising solution to mitigate the potential privacy leakage of sensitive information about individuals (since the data stays locally with each client), such as typing history, shopping transactions, geographical locations, medical records, and etc. However, recent works have demonstrated that FL may not always provide sufficient \emph{privacy} and \emph{robustness} guarantees.  In terms of privacy leakage, exchanging the model updates throughout the training process can still reveal sensitive information \cite{bhowmick2018protection,melis2019exploiting} and cause deep leakage such as  pixel-wise accurate image recovery  \cite{zhu2019deep,yin2021see}, either to a third-party (including other participating clients) or the central server. In terms of robustness, FL systems are vulnerable to model poisoning attacks, where the attacker controls a subset of (malicious) clients, aiming  to either prevent
the convergence of the global model (a.k.a. Byzantine attacks) \cite{fang2020local,baruch2019little}, or implant a backdoor trigger into the global model to cause targeted misclassification (a.k.a. backdoor attacks)  \cite{bagdasaryan2020backdoor,wang2020attack}. 

To mitigate the privacy leakage in FL, Differential Privacy (DP) \cite{dwork2006calibrating,dwork2014algorithmic} has been adopted  as a rigorous privacy notion. Several existing frameworks  \cite{mcmahan2018learning,geyer2017differentially,li2020differentially} applied DP in FL to provide \emph{client-level} privacy under the assumption of a trusted server: whether a client has participated in the training process cannot be inferred from the released model,  and the client's whole dataset remains private. Other works in FL \cite{zheng2021federated,li2020differentially,xu2019hybridalpha,truex2019hybrid} focused on \emph{record-level} privacy: whether a data record  has participated during training can not be inferred by adversaries, including the server which may be untrusted. Record-level privacy is more relevant in cross-silo (as versus cross-device) FL scenarios, such as multiple hospitals collaboratively learn a prediction model for COVID-19, in which case what needs to be protected is the privacy of each patient (corresponding to each record in a hospital's dataset).



In this paper, we focus on cross-silo FL with \emph{record-level} DP, where each client possesses a set of raw records (which are aggregated from some individuals), and each record corresponds to an individual's private data. Existing solutions in this setting can be categorized as Centralized DP (CDP), Local DP (LDP)\footnote{Following \cite{lyu2020privacy,naseri2020toward,truex2019hybrid}, we use LDP to refer to the client based approaches for ease of presentation, but it is different from the traditional LDP for data collection in \cite{duchi2013local,wang2017locally,gu2020pckv}.}, and Distributed DP (DDP). In CDP-based solutions \cite{mcmahan2018learning,geyer2017differentially}, each client submits the raw update to a \emph{trusted} server who applies a model aggregation mechanism with DP guarantees. In LDP-based solutions \cite{li2020differentially,pihur2018differentially}, each client adds noise to the update with DP guarantees before sending it to the server, where the server and other clients are assumed as \emph{untrusted}. Though relying on a weaker trust assumption, LDP approaches suffer from poor utility because the noise added by all the clients are accumulated when the server aggregates all updates (as verses in CDP only one party, the server, adds noise).  In DDP-based solutions \cite{truex2019hybrid,xu2019hybridalpha,lyu2020lightweight}, each client adds \emph{partial} noise to achieve the global DP noise as in server based (i.e., the required noise are jointly added by all clients in a \emph{distributed} way), and sends the encrypted output to the server. The utilized cryptographic primitives, such as additive homomorphic encryption (HE), guarantee that all but the final result are hidden from the server and other clients. By leveraging cryptography techniques, DDP-based solutions avoid placing trust in the server (instead with an honest-but-curious assumption), and offers better utility than LDP-based solutions. However, the utility enhancement of DDP-based approaches is sensitive to the minimum number of trusted parties, which should be known in advance to derive the required noise amount, and it reduces to LDP in the worst case when all-but-one clients collude to infer the victim client's data. Furthermore, such encryption based methods prohibits the server from auditing clients' model updates, which leaves room for malicious attacks. For example, malicious clients can  introduce stealthy backdoor functions into the global model without being detected.

On the robustness of FL, recent works \cite{sun2019can,naseri2020toward} empirically observed that the noise added in the clipped gradients under CDP and/or LDP is able to defend against backdoor attacks. The intuition is that CDP limits the information learned about a specific client, while LDP does so for records in a client’s dataset. In both cases, the impact of poisoned data will be reduced, while simultaneously providing DP guarantees. However, CDP assumes a trusted server, and LDP defends against backdoor attacks only when the corrupted (malicious) clients also implement the noise augmentation mechanism. If  malicious clients opt out from the LDP protocol \cite{naseri2020toward}, the framework becomes less robust,  because the malicious clients can potentially have a bigger impact on the aggregated model when other benign clients honestly add noise (thus have less impact) to satisfy LDP. In other words, the noise of LDP actually reduces robustness in practice where the attacker has full control on the corrupted clients. Also, a concurrent work \cite{guerraoui2021differential} shows that (in both theoretical and experimental perspectives) the classical approaches to Byzantine-resilience and LDP are practically incompatible in machine learning. Furthermore, similar observations of manipulation vulnerability of LDP have been made in \cite{cheu2021manipulation,cao2021data} under the application of data aggregation for statistic queries.

In summary, neither CDP, LDP, nor DDP based solutions can achieve both privacy and robustness without sacrificing performance, where CDP does not protect privacy against the server, LDP has poor utility and is vulnerable to a strong attacker who has full control over malicious clients, and DDP is not able to audit (malicious) clients' updates. The main challenge lies in the dilemma between the server learning little information from  clients' data (for privacy), while being able to detect anomalous submissions injected by malicious clients (for robustness).

In this paper, we propose a \textbf{P}rivacy and \textbf{R}obustness \textbf{E}nhanced FL framework via \textbf{C}rypto-\textbf{A}ided \textbf{D}P (\textbf{PRECAD}), which is the first scheme that  simultaneously provides strong DP guarantees and \emph{quantifiable} robustness against model poisoning attacks, while providing good model utility. This is achieved by combining DP with secure multi-party computation (MPC) techniques (secret sharing). PRECAD involves two honest-but-curious and non-colluding servers. This setting has been widely formalized and instantiated in previous works such as \cite{mohassel2017secureml,roy2020crypt,corrigan2017prio,agrawal2019quotient,he2020secure}. In PRECAD, each client implements clipping of gradients (which will be the local updates) in both record-level and client-level, then uploads the random shares of its local update to two servers respectively. The two servers run a secret-sharing based MPC protocol to securely verify \emph{client-level} clipping, and jointly adds Gaussian noise to the aggregated result. Our protocol guarantees that the two servers can only learn the validity of client-submitted model update  (whether it is bounded by the clipping norm, which leaks nothing about a benign client's data), and the result of the noisy aggregation (where the privacy leakage at record-level is carefully accounted by DP). PRECAD simulates the CDP paradigm (a.k.a. SIM-CDP \cite{mironov2009computational}) in FL, but does not rely on the assumption of trusted server(s). Furthermore, due to the verifiable client-level clipping, the contribution of each (malicious) client to the global model is bounded. We use DP to show the bounds of robustness guarantee, where the noise in PRECAD for privacy purpose \emph{enhances} robustness (in contrast, noise in LDP \emph{reduces} robustness). A comparison among the different privacy  models/frameworks, and their privacy/robustness guarantees (including CDP, LDP, DDP and PRECAD)  is shown in Table \ref{tab:DP_compare}. 

\textbf{Contributions.}
Our main contributions include:

1) The proposed scheme PRECAD is the first framework that simultaneously enhances privacy-utility tradeoff of DP and robustness against model poisoning attacks for FL under a practical threat model: for private data inference, we assume servers are honest-but-curious and allow the collusion of all clients except the victim; for model poisoning attacks, the attackers has full control over the corrupted clients. In PRECAD, the record-level clipping at each client, combining with secret sharing and server perturbation, ensures record-level privacy, while the client-level clipping and secure verification ensure robustness against malicious clients.

2) We show that PRECAD satisfies record-level DP in Theorem \ref{thm:privacy_analysis} while providing quantifiable robustness against model poisoning attacks in Theorem \ref{thm:robustness_analysis}, where the theoretical analysis shows that the noise for privacy purpose enhances the robustness guarantees.

3) We conduct several experiments to demonstrate PRECAD's enhancement on the privacy-utility tradeoff (compared with LDP-based approaches) and robustness against backdoor attacks (which is a prominent type of targeted model-poisoning attacks). The experimental results validate our theoretical analysis that DP-noise in PRECAD enhances model robustness.

\section{Preliminaries}

\subsection{Differential Privacy (DP)}
\label{sec:preliminaries_DP}

Differential Privacy (DP) is a rigorous mathematical framework for the release of information derived from private data. Applied to machine learning, a differentially private training mechanism allows the public release of model parameters with a strong privacy guarantee: adversaries are   limited in what they can learn about the original training data based on analyzing the parameters, even when they have access to arbitrary side information. The formal definition is as follows:
\begin{definition}[$(\epsilon,\delta)$-DP \cite{dwork2014algorithmic,dwork2006calibrating}]
\label{def:DP}
    For $\epsilon\in[0,\infty)$ and $\delta\in[0,1)$, a randomized mechanism $\mathcal{M}:\mathcal{D}\rightarrow \mathcal{R}$ with a domain $\mathcal{D}$ (e.g., possible training datasets) and range $\mathcal{R}$ (e.g., all possible trained models) satisfies $(\epsilon,\delta)$-Differential Privacy (DP) if for any two neighboring datasets $D,D^\prime\in\mathcal{D}$  and for any subset of outputs $S\subseteq \mathcal{R}$, it holds that
    \begin{align*}
        \mathbb{P}[\mathcal{M}(D)\in S]\leqslant e^\epsilon\cdot \mathbb{P}[\mathcal{M}(D^\prime)\in S] +\delta
    \end{align*}
    where a larger $\epsilon$ and $\delta$ indicate a less private  mechanism.
\end{definition}

\textbf{Gaussian Mechanism.} 
A common paradigm for approximating a deterministic real-valued function $f:\mathcal{D}\rightarrow\mathbb{R}$ with a differentially private
mechanism is via additive noise calibrated to $f$'s sensitivity
$s_f$, which is defined as the maximum of the absolute distance $|f(D)-f(D^\prime)|$, where $D$ and $D^\prime$ are neighboring datasets. The Gaussian Mechanism is defined by $\mathcal{M}(D)=f(D)+\mathcal{N}(0,s_f^2\cdot\sigma^2)$, where $\mathcal{N}(0,s_f^2\cdot\sigma^2)$ is the normal (Gaussian) distribution with mean 0 and standard deviation $s_f\sigma$. It was shown that the mechanism $\mathcal{M}$ satisfies $(\epsilon,\delta)$-DP if $\delta\geqslant\frac{4}{5}e^{-(\sigma\epsilon)^2/2}$ and $\epsilon<1$ \cite{dwork2014algorithmic}. Note that we use an advanced privacy analysis tool in \cite{dong2019gaussian} (refer to Lemma \ref{lem:Gaussian_mechanism_GDP} in Appendix \ref{apx:GDP}), which works for all $\epsilon>0$. 

\textbf{DP-SGD Algorithm.}
The most well-known differentially private algorithm in machine learning is DP-SGD \cite{abadi2016deep}, which introduces two modifications to the vanilla stochastic gradient descent (SGD). First, a \emph{clipping step} is applied to the gradient so that the gradient is in effect bounded. This step is necessary to have a finite sensitivity. The second modification is \emph{Gaussian noise augmentation} on the summation of clipped gradients, which is equivalent to applying the Gaussian mechanism to the updated iterates. The privacy accountant of DP-SGD is shown in Appendix \ref{apx:GDP}.

\subsection{FL with DP}
\label{sec:preliminaries_FL_DP}

FL is a collaborative learning setting to train machine learning models. It involves multiple clients, each holding their own private dataset, and a central server (or aggregator). Unlike the traditional centralized approach, data is not stored at a central server; instead, clients train models locally and exchange updated parameters with the server, who aggregates the received local model parameters and sends them to the clients. FL involves multiple iterations. At each iteration, the server randomly chooses a subset of clients and sends them the current model parameter; then these clients locally compute  training gradients according to their local datasets and send  the updated parameters to the server. The latter aggregates the results and updates the global model. After a certain number of iterations (or until convergence), the final model parameter is returned as the output of the FL process.

\textbf{Record-level v.s. Client-level DP.}
In FL, the \emph{neighboring datasets} $D$ and $D^\prime$ (in Definition \ref{def:DP}) can be defined at two distinct levels:  \emph{record-level} and \emph{client-level}. For record-level DP, $D$ and $D^\prime$ are defined to be neighboring if $D^\prime$ can be formed by adding or removing a single training record/example from $D$. On the other hand, for client-level DP, $D^\prime$ is obtained by adding or removing one client's whole training dataset from $D$. In this paper, the main privacy goal is to guarantee \emph{record-level} DP, as this is most relevant in the FL applications we consider. On the other hand, \emph{client-level} DP  is also achieved as a by-product of our protocol, and we will utilize it to show the robustness against model poisoning attacks.

\subsection{Secret Sharing}
\label{sec:preliminaries_secret_sharing}
We will make use of the additive secret sharing primitive in \cite{cramer2005share,corrigan2017prio}.  
Consider a decentralized setting with $n$ clients and $s$ servers, each client holds an integer $x_i$ and the servers want to compute the sum of the clients' private values $\sum_{i=1}^n x_i$. All arithmetic of a secret sharing scheme takes place in a finite field $\mathbb{F}$ with a public, large prime $\mathsf{p}$. For convenience, we use $c=a+b\in\mathbb{F}$ to indicate $c=a+b \mod \mathsf{p}$. The additive secret sharing for computing sums proceeds in three steps:

Step 1: Upload. Each client $i$ splits its private value $x_i$ into $s$ shares, one per server, using a secret-sharing scheme. In particular, the client picks random integers $[x_i]_1,\cdots,[x_i]_s\in\mathbb{F}$, subject to the constraint: $x_i=\sum_{j=1}^s [x_i]_j \in\mathbb{F}$. The client then sends one share of its submission to each server through a secure (private and authenticated) channel.

Step 2: Aggregate. Each server $j$ aggregates all received shares from $n$ clients by computing the value of an accumulator $A_j=\sum_{i=1}^n [x_i]_j \in\mathbb{F}$.

Step 3: Publish. Once the servers have received shares from all clients, they publish their accumulator values. Computing the sum of the accumulator values $\sum_{j=1}^s A_j\in\mathbb{F}$  yields the desired sum $\sum_{i=1}^n x_i$ of the clients' private values, as long as the modulus $\mathsf{p}$ is larger than the final result  (i.e., the sum $\sum_{i=1}^n x_i$ does not \emph{overflow} the modulus).

The above secret sharing scheme protects clients' privacy in an unconditional and information-theoretic sense: an attacker who gets hold of any subset of up to $s-1$ shares of $x_i$ (i.e., at least one of the servers is honest) learns nothing, except what the aggregate statistic $\sum_{i=1}^n x_i$ itself reveals.


\section{Problem Statement}

\subsection{System Model}
\label{sec:system_model}

We assume multiple parties in the FL system: two aggregation servers  ($\mathsf{S_A}$ and $\mathsf{S_B}$) and $n$ participating clients $\{\mathsf{C}_1,\cdots,\mathsf{C}_n\}$.  The servers hold a global model $\theta_t\in\mathbb{R}^d$  and each client $\mathsf{C}_i$ possesses a private training dataset $D_i$. Each server communicates with the other server and each client through a secure (private and authenticated) channel. At the $t$-th iteration ($t=0,1,\cdots,T$), the servers randomly select a subset of clients $\mathcal{I}_t\subseteq\{1,\cdots,n\}$ and send the current global model parameter $\theta_t$ to them. Next, each client $\mathsf{C}_i~(i\in\mathcal{I}_t)$ who receives global $\theta_t$ trains the local model $\theta_t^i$ from its own private dataset $D_i$ and sends the update $\Delta\theta_t^i=\theta_t^i-\theta_t$ (i.e., the difference between the local model and the global model) to the servers. Then, the servers update the global model by aggregating all $\Delta \theta_t^i$:
\begin{align}
\label{equ:system_model}
	\theta_{t+1} = \theta_t + \eta\cdot\sum\nolimits_{i\in \mathcal{I}_t} w_t^i\Delta \theta_t^i,
\end{align}
where $\eta$ is the learning rate of the global model and $w_t^i$ is the aggregation weight of client $\mathsf{C}_i$ at the $t$-th iteration. They  will keep iterating the above procedure until convergence.

Note that the model parameter $\theta_t$ is a $d$-dimensional real vector (i.e., $\theta_t\in\mathbb{R}^d$), while the utilized secret sharing scheme for secure aggregation is taken in a finite field $\mathbb{F}^d$. For any real value $a\in\mathbb{R}$ (or a real vector $a\in\mathbb{R}^d$),  we can embed the real values into a finite field $\mathbb{F}$ (with size $\mathsf{p}$) using a \emph{fixed-point} representation, as long as the size of the field is large enough to avoid overflow. In this paper, we use $[a]_\mathbb{F}$ to denote the fixed-point representation (within a finite field $\mathbb{F}$) of a real value $a$ and use $[a]_\mathsf{A}$ to denote the share of $[a]_\mathbb{F}$ held by $\mathsf{S_A}$, where $[a]_\mathbb{F},[a]_\mathsf{A},[a]_\mathsf{B}\in\mathbb{F}$ and $[a]_\mathbb{F}=[a]_\mathsf{A}+[a]_\mathsf{B}\mod\mathsf{p}$.

\subsection{Threat Model}
\label{sec:threat_model}

In this paper, we consider two types of attacks: record inference attacks and model poisoning attacks, where the first one compromises data privacy and the second one compromises the model robustness. The threat models of them are different:

\textbf{Record Inference Attacks.} 
To infer one record of a benign client (who is the victim), the attacker can corrupt  \emph{at most one} server and a subset of clients (except the victim). We assume that the two servers are \emph{honest-but-curious} and \emph{non-colluding}. This setting has been widely formalized and instantiated in previous works such as \cite{mohassel2017secureml,roy2020crypt,corrigan2017prio,agrawal2019quotient,he2020secure}. Non-colluding means that they avoid revealing information to each other beyond what is allowed by the protocol definition. Honest-but-curious (a.k.a. semi-honest) means that they follow the protocol instructions honestly, but will try to learn additional information. We assume the corrupted clients are also honest-but-curious, and can communicate with the corrupted server. We note that even when the corrupted clients deviate from the protocol, they do not obtain additional gain for privacy inference because their submissions only affect the parameter of the global model but not the privacy protocol.

\textbf{Model Poisoning Attacks.} 
The attacker can corrupt a subset of malicious clients (but not the servers) to implement the model poisoning attacks. We assume the attacker has full control on both the local training data and the submission to the servers over these corrupted (malicious) clients, but has no influence on other benign clients. Furthermore, the malicious clients can fully cooperate with each other to achieve a stronger attack  influence. We assume that the number of malicious clients is less than benign clients. Following \cite{bagdasaryan2020backdoor,sun2019can,naseri2020toward}, we consider a \emph{model-replacement} methodology for model poisoning attacks.  In FL (described in Sec. \ref{sec:system_model}), each selected client $\mathsf{C}_{i}(i\in\mathcal{I}_t)$ sends the update $\Delta\theta_t^i=\theta_t^i-\theta_t$ to the servers. At $t$-th iteration, we assume only one malicious client (say client $\mathsf{C}_{i^*}$) is selected, then he/she attempts to replace the global model by a targeted model $\theta^*$ via sending $\Delta \theta_t^{i^*} = \frac{1}{\eta w_t^{i^*}}(\theta^*-\theta_t)$. We note that each local model $\theta_t^i$ may be far from the global model $\theta_t$. However, as the global model converges, these deviations start to cancel out, i.e., $\sum_{i\in \mathcal{I}_t, i\neq i^*} w_t^i\Delta \theta_t^i=\sum_{i\in \mathcal{I}_t, i\neq i^*} w_t^i(\theta_t^i-\theta_t) \approx 0$. Therefore, if we assume the model has sufficiently converged, the parameter of the global model in \eqref{equ:system_model} will be replaced by $\theta_{t+1} = \theta^* + \eta\sum_{i\in \mathcal{I}_t, i\neq i^*} w_t^i\Delta \theta_t^i~~\approx \theta^* $. When multiple malicious clients appear in the same iteration, we assume that they can coordinate with each other and divide the malicious update $\Delta \theta_t^{i^*}$ evenly. Furthermore, such attack can be implemented with multiple iterations.

\section{Our Framework: PRECAD}

In this section, we introduce the proposed framework called PRECAD. It simulates Centralized DP (CDP) in FL, without relying on the assumption of trusted server(s), and provides quantifiable robustness against malicious clients.

\label{sec:framework}
\begin{figure}[!t]
    \centering
    \includegraphics[width=3.4in]{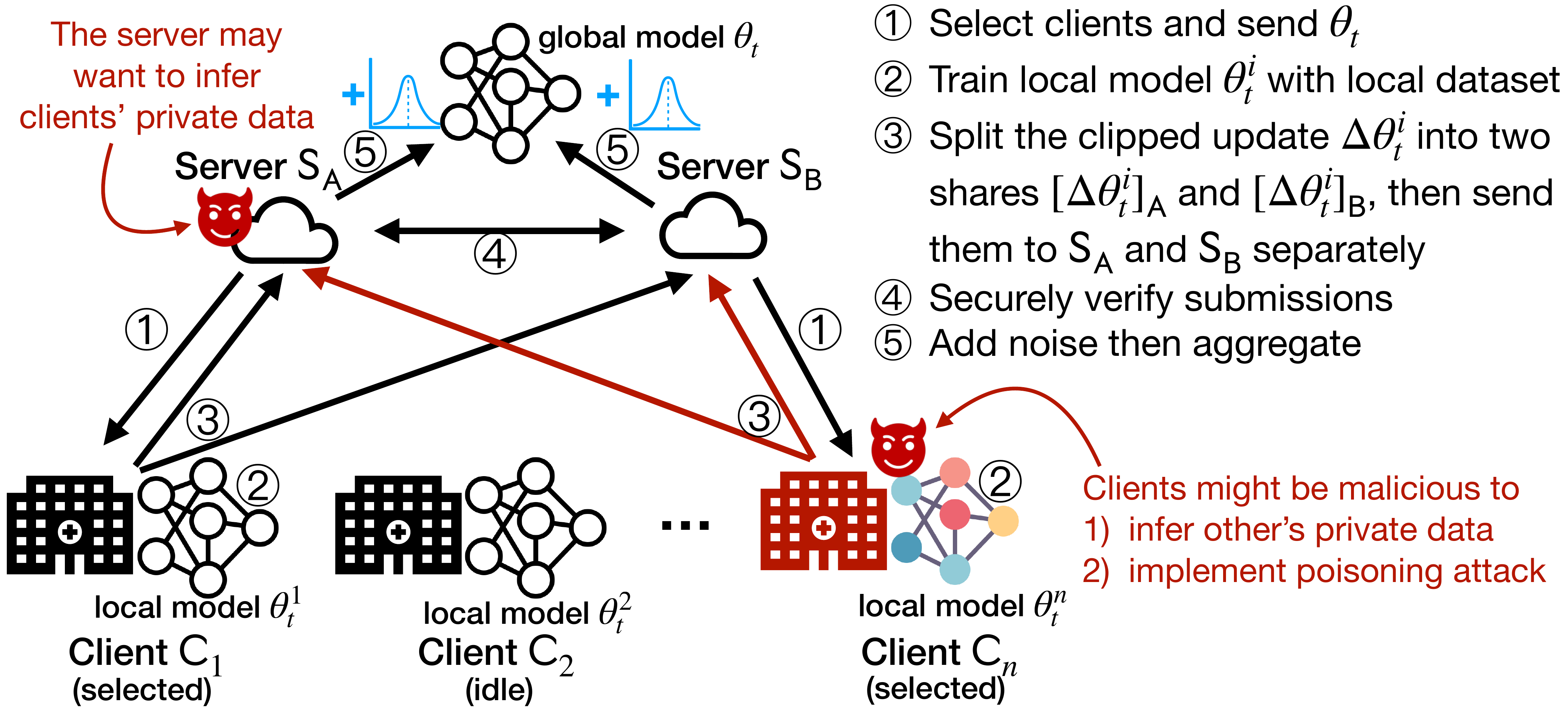}
    \vspace{-6mm}
    \caption{Illustration of our  PRECAD framework.}
    \label{fig:framework}
\end{figure}

\subsection{Overview}
An illustration of our framework PRECAD is shown in Figure \ref{fig:framework}. It follows the general FL process (as discussed in Sec. \ref{sec:system_model}) with secure aggregation via additive secret sharing (refer to Sec. \ref{sec:preliminaries_secret_sharing} for preliminaries of secret sharing). In PRECAD, two servers engage in an interactive protocol to add Gaussian noise to the aggregated model updates,  using a secure aggregation method. Note that other MPC primitives are possible, but we opt to use additive secret sharing in this paper. In each iteration, the model update of each client is supposed to be clipped before submission, and the servers run another MPC protocol based on the submitted shares, to verify the validity of each submission (ensure the norm is smaller than a threshold), in order to mitigate poisoning attacks. Since the noise for DP is added by the servers (rather than by the clients) after clients' submissions, we will show that (in Sec. \ref{sec:robustness_analysis}) it can actually enhance model robustness.

\begin{algorithm}[!t]
    \small
	\caption{PRECAD}
	\begin{algorithmic}[1]
		\REQUIRE Initialization $\theta_0\in\mathbb{R}^d$, size of dataset $|D_i|$ (for $i=1,\cdots,n$), record-level clipping norm bound $R$, client-level clipping norm bound $C$, learning rate $\eta$ of the global model.
		\FOR{$t=0,1,\cdots,T-1$}
		\STATE Select a subset of clients $\mathcal{I}_t\subseteq\{1,\cdots,n\}$ , where each client is selected with probability $q$.
		\STATE Each client $\mathsf{C}_i~(i\in\mathcal{I}_t)$ implements the local updates via Algorithm \ref{alg:local_updates} and sends shares $[\Delta\theta_t^i]_\mathsf{A}$ to server $\mathsf{S_A}$  and  $[\Delta\theta_t^i]_\mathsf{B}$ to server $\mathsf{S_B}$.
		\STATE $\mathsf{S_A}$ and $\mathsf{S_B}$ securely verify whether $\|\Delta\theta_t^i\|_2\leqslant C$ holds via secure validation and identify the valid clients $\mathcal{I}_t^*\coloneqq \{i\in\mathcal{I}_t: \|\Delta\theta_t^i\|_2\leqslant C\}$ without learning anything else.
		\STATE $\mathsf{S_A}$ and $\mathsf{S_B}$ independently draw random Gaussian noise $\xi_t^\mathsf{A},\xi_t^\mathsf{B}\sim\mathcal{N}(0,(R\sigma)^2\cdot \mathbf{I})$, compute the noisy shares $[\sum_{i\in\mathcal{I}_t^*}\Delta\theta_t^i]_\mathsf{A}+[\xi_t^\mathsf{A}]_\mathbb{F}$ and $[\sum_{i\in\mathcal{I}_t^*}\Delta\theta_t^i]_\mathsf{B}+[\xi_t^\mathsf{B}]_\mathbb{F}$ respectively and sends the result to another server.
		\STATE Servers obtain the updated global model $\theta_{t+1}$ by \eqref{equ:global_model_update}.
		\ENDFOR
		\ENSURE The final model parameter $\theta_T$.
	\end{algorithmic}
	\label{alg:framework}
\end{algorithm}

\textbf{Main Steps.}
Our instantiation of PRECAD is shown in Algorithm \ref{alg:framework}. In each iteration $t$, the servers and clients execute the following four steps:

Step 1. Selection of Participating Clients. The servers select a subset of clients to participate in the current iteration. Since both servers follow the protocol honestly, either server $\mathsf{S_A}$ or $\mathsf{S_B}$ can execute this step, where each client is randomly selected with probability $q$. Then, servers send the current model parameter $\theta_t\in\mathbb{R}^d$ to these clients and wait for responses. 

Step 2. Local Model Update and Submission. After receiving the request from the servers, each client $\mathsf{C}_i$ trains the local model with the private local dataset $D_i$, where the gradient per record is clipped by $R$ (\emph{record-level} clipping). Then,  the update is clipped by $C$ (\emph{client-level} clipping) and the clipped update  $\Delta\theta_t^i\in\mathbb{R}^d$ is split into two shares $[\Delta\theta_t^i]_\mathsf{A},[\Delta\theta_t^i]_\mathsf{B}\in\mathbb{F}^d$ (see more details in Sec. \ref{sec:local_update}), which are sent to server $\mathsf{S_A}$ and $\mathsf{S_B}$ respectively.

Step 3. Secure Submission Verification. Since a malicious client may   send a large submission (with $\ell_2$-norm beyond the bound $C$), the servers, who hold shares $[\Delta\theta_t^i]_\mathsf{A}$ and $[\Delta\theta_t^i]_\mathsf{B}$ respectively, must   securely verify whether the $\ell_2$-norm of each $\Delta\theta_t^i$ is indeed bounded by $C$ (refer to Sec. \ref{sec:secure_validation} for the detailed protocol). Note that, this verification step, which outputs either valid (the servers accept $\Delta\theta_t^i$) or invalid (the servers reject $\Delta\theta_t^i$), leaks nothing about a \emph{benign} client's private information because benign clients follow the protocol honestly  and their local model update submissions will always be valid.

Step 4. Secure Aggregation with Noise. After verifying all submissions, server $\mathsf{S_A}$ draws a real random vector $\xi_t^\mathsf{A}\sim\mathcal{N}(0,(R\sigma)^2\cdot \mathbf{I})$ and converts it into the fixed-point representation $[\xi_t^\mathsf{A}]_\mathbb{F}$, where $\mathbf{I}$ is the identity matrix with size $d$. Denote the set of  indices of \emph{valid} submissions as $\mathcal{I}_t^*$. Then,  $\mathsf{S_A}$ aggregates all valid shares  with Gaussian noise by computing $[\sum_{i\in\mathcal{I}_t^*}\Delta\theta_t^i]_\mathsf{A}+[\xi_t^\mathsf{A}]_\mathbb{F}\in\mathbb{F}^d$. Similarly, $\mathsf{S_B}$ computes $[\sum_{i\in\mathcal{I}_t^*}\Delta\theta_t^i]_\mathsf{B}+[\xi_t^\mathsf{B}]_\mathbb{F}\in\mathbb{F}^d$. By exchanging the above  with each other, both servers obtain the sum (modulo $\mathsf{p}$) and convert it into a real vector $\sum\nolimits_{i\in\mathcal{I}_t^*}\Delta\theta_t^i+\xi_t^\mathsf{A}+\xi_t^\mathsf{B}\in\mathbb{R}^d$, which is utilized to update the global model parameter:
\begin{align}
\label{equ:global_model_update}
	\theta_{t+1} = \theta_t+\frac{\eta}{\sum_{i\in\mathcal{I}_t^*}p_i|D_i|}\left(\sum\nolimits_{i\in\mathcal{I}_t^*}\Delta\theta_t^i+\xi_t^\mathsf{A}+\xi_t^\mathsf{B}\right),
\end{align}
where $\eta$ is the learning rate of the global model.

\textbf{Comparisons.}
Our framework provides record-level DP and robustness against model poisoning attacks. The detailed analysis of privacy and robustness is provided in Sec. \ref{sec:privacy_robustness_analysis}. Comparing with the LDP-based solutions \cite{li2020differentially,pihur2018differentially} (where each client adds noise during the process of local training and sends the result to the server in clear) and DDP-based solutions \cite{truex2019hybrid,xu2019hybridalpha} (where each submission with noise augmentation is encrypted, and the central server only observes the \emph{aggregation} of noisy submissions), the proposed framework PRECAD has several advantages (also refer to Table \ref{tab:DP_compare}):

1) Better utility under the same privacy budget of DP. In the LDP-based solutions, the noisy submissions of all clients are aggregated by the server, thus the global model is updated with accumulated noise, where the variance is proportional to the number of participating clients. Although the DDP-based solutions can reduce the noise during local training (due to encrypted submissions),  the  variance of added noise is proportional to the ratio between the number of all clients and the minimum number of non-colluding clients (the accurate value is difficult to obtain in real-world applications). Furthermore, they  reduce to the LDP-based solutions in the worst case when all other clients collude to refer the victim's data. However, the noise in PRECAD only accumulates twice (by $\mathsf{S_A}$ and $\mathsf{S_B}$) and is independent of the number of non-colluding clients. 

2) Effective defense against poisoning attacks. In FL, it is not easy to perform anomaly detection because the server cannot access clients' private datasets. This task is even more challenging with DP,  since the random noise added by clients decreases the distributional distance between normal and abnormal submissions (thus making it harder to distinguish them). Therefore, both LDP-based and DDP-based solutions are vulnerable to malicious submissions, where the LDP-based solutions require more noise under the same privacy guarantee, and the encryption of submissions in DDP-based solutions limits the capability of anomaly detection. In contrast, the secure validation step in PRECAD (which learns nothing about a benign client) can completely bound the norm of malicious submissions at client-level (those who violate this constraint will be detected), and the noise added by servers increases the uncertainty when the attackers attempt to  modify the global model.

\subsection{Local Model Update}
\label{sec:local_update}

\begin{algorithm}[t!]
    \small
	\caption{Local update of client $\mathsf{C}_i$ at $t$-th iteration}
	\begin{algorithmic}[1]
		\REQUIRE Current global model $\theta_t$, local dataset $D_i$, record-level clipping bound $R$, client-level clipping bound $C$.
		\STATE Sample a set of indices $\mathcal{S}\subseteq\{1,\cdots,|D_i|\}$, where each index is sampled with probability $p_i$.
		\FOR{$s\in\mathcal{S}$}
		\STATE $g_s=\nabla_\theta \ell(\theta_t,z_s)$  \verb| //gradient per record|
		\STATE $\tilde{g}_s= g_s\cdot\min\{1,~R/\|g_s\|_2\}$  \verb| //record-level clipping|
		\ENDFOR
		\STATE $\Delta\tilde{\theta}_t^i=-\sum_{s\in\mathcal{S}}\tilde{g}_s$
		\STATE $\Delta\theta_t^i= \Delta\tilde{\theta}_t^i\cdot\min\{1,~C/\|\Delta\tilde{\theta}_t^i\|_2\}$  \verb| //client-level clipping|
		\STATE Convert $\Delta\theta_t^i\in\mathbb{R}^d$ to $[\Delta\theta_t^i]_\mathbb{F}\in\mathbb{F}^d$ and split it into two shares $[\Delta\theta_t^i]_\mathsf{A},[\Delta\theta_t^i]_\mathsf{B}\in\mathbb{F}^d$.
		\ENSURE Send $[\Delta\theta_t^i]_\mathsf{A}$ to server $\mathsf{S_A}$  and  send $[\Delta\theta_t^i]_\mathsf{B}$ to server $\mathsf{S_B}$.
	\end{algorithmic}
	\label{alg:local_updates}
\end{algorithm}

The protocol of local model update is shown in Algorithm \ref{alg:local_updates}. After receiving the current global model parameter $\theta_t$,  client $\mathsf{C}_i$ samples a subset of records from the local dataset, where each record is sampled with probability $p_i$. For each sampled record, the corresponding gradient is computed and then clipped with $R$ (\emph{record-level} clipping). Then, $\mathsf{C}_i$ can compute the sum of the \emph{negative} gradients (for gradient descent) and then clip the result by $C$ (\emph{client-level} clipping). We denote the clipped result as $\Delta\theta_t^i\in\mathbb{R}^d$, which is converted to its \emph{fixed-point} representation $[\Delta\theta_t^i]_\mathbb{F}\in\mathbb{F}^d$ and then split into two shares in $\mathbb{F}^d$, which are sent to two servers respectively.

The record-level clipping guarantees that by removing or adding one record from client $\mathsf{C}_i$'s dataset, the aggregation result at the server-side changes at most $R$ in terms of $\ell_2$-norm (i.e., bounded sensitivity), thus adding Gaussian noise on the aggregation provides record-level DP (shown in Sec. \ref{sec:privacy_analysis}). Similarly, the client-level clipping guarantees client-level DP on the noisy aggregation. Though client-level DP is not our privacy goal, we will show that (in Sec. \ref{sec:robustness_analysis}) it can be exploited to provide robustness of the learning process against model poisoning attacks. While the record-level clipping can automatically achieve some client-level clipping, the exact clipping bound depends on the number of sampled records (which is a random variable).  We use explicit client-level clipping to ensure the robustness is controllable.

For malicious clients who may deviate from the protocol execution and collude with other clients, the \emph{record-level} privacy is not guaranteed because they can opt out  record-level clipping. However, even under the presence of malicious clients, the aggregation in the server-side guarantees \emph{client-level} privacy because the submissions without the client-level clipping (i.e., the $\ell_2$ norm of real vector $\Delta\theta_t^i$ exceeds the bound $C$) will be rejected by the servers during secure validation (refer to Sec. \ref{sec:secure_validation}). Thus, malicious clients have to execute client-level clipping; otherwise, their submissions will be rejected and then have no influence on the aggregation result.

\subsection{Secure Validation}
\label{sec:secure_validation}
After a client $\mathsf{C}_i$ submits its shares of $\Delta\theta_t^i$, the servers need to securely verify whether $\|\Delta\theta_t^i\|_2\leqslant C$ holds. For the ease of presentation, we use $x$ to denote $\Delta\theta_t^i$ in this subsection. Specifically, server $\mathsf{S_A}$ and $\mathsf{S_B}$ hold shares $[x]_\mathsf{A}$ and $[x]_\mathsf{B}$ separately and want to verify whether the value of  $\mathsf{Valid}(x)\coloneqq\mathbbm{1}(x^\top x- C^2\leqslant0)$ is 1 without leaking any additional information, where $\mathbbm{1}(\cdot)$ denotes the indicator function and  $C^2\in\mathbb{F}$ is the fixed-point representation of the squared client-level clipping bound (which is public). To do so, we build upon the secure multiplication technique of Beaver \cite{beaver1991efficient}, which computes the multiplication of two secret-shared numbers (refer to Appendix \ref{apx:Beaver_MPC} for implementation details). The original Beaver's protocol is used to compute the shares of multiplication result of two private numbers, and we can extend it to the case of the inner product $x^\top x$ as follows.

We assume that $\mathsf{S_A}$ and $\mathsf{S_B}$ have access to a sufficient number of random one-time-use shares of $a\in\mathbb{F}^d$ and $r\in\mathbb{F}$ with the constraint $a^\top a=r$, where $a$ and $r$ correspond to Beaver's triples in the original Beaver's multiplication protocol \cite{beaver1991efficient}. Similar to Beaver's triples, the random shares of $a$ and $r$ can be provided by a trusted third party (TTP), or are generated off-line via cryptography techniques, such as additive homomorphic encryption \cite{keller2018overdrive} or oblivious transfer \cite{keller2016mascot}. Since their security and performance have been demonstrated, we assume these random shares are ready for use at the initial step of our protocol. To securely compute the result of $\mathsf{Valid}(x)$, the two severs implement the following steps:

Step 1. Server $\mathsf{S_A}$ computes $[b]_\mathsf{A}=[x]_\mathsf{A}-[a]_\mathsf{A}$, and server $\mathsf{S_B}$ computes $[b]_\mathsf{B}=[x]_\mathsf{B}-[a]_\mathsf{B}$. Then, they exchange the value of $[b]_\mathsf{A}$ and $[b]_\mathsf{B}$ with each other, and both of them hold the value of the vector $b=x-a\in\mathbb{F}^d$.

Step 2. Server $\mathsf{S_A}$ computes $y_\mathsf{A} = [r]_\mathsf{A}+2b^\top[a]_\mathsf{A}+(b^\top b-[C^2]_\mathbb{F})/2$, and server $\mathsf{S_B}$ computes $y_\mathsf{B} = [r]_\mathsf{B}+2b^\top[a]_\mathsf{B}+(b^\top b-[C^2]_\mathbb{F})/2$, where the addition and division are computed in the field $\mathbb{F}$ (or $\mathbb{F}^d$ for vectors). Now, the servers hold the shares of $y$, which essentially has the following representation (which is taken in the field $\mathbb{F}$)
\begin{align*}
    y&=y_\mathsf{A} + y_\mathsf{B} \\
    &= [r]_\mathsf{A} + [r]_\mathsf{B} + 2b^\top([a]_\mathsf{A}+[a]_\mathsf{B})+(b^\top b-C^2)\\
    &= a^\top a + 2b^\top a +b^\top b-C^2\\
    &= (x-b)^\top (x-b) + 2b^\top (x-b) +b^\top b-C^2\\
    &=x^\top x-C^2,
\end{align*}
which indicates $\mathsf{Valid}(x)=\mathbbm{1}(x^\top x- C^2\leqslant0)=\mathbbm{1}(y\leqslant0)$. 

Step 3. The servers can compute the shares $[\mathbbm{1}(y\leqslant 0)]_\mathsf{A}$ and $[\mathbbm{1}(y\leqslant 0)]_\mathsf{B}$ via the comparison gate (refer to \cite{knott2021crypten} for more details). Finally, by exchanging the above two shares with each other, both servers can compute the value $\mathsf{Valid}(x)=\mathbbm{1}(y\leqslant0)=[\mathbbm{1}(y\leqslant 0)]_\mathsf{A}+[\mathbbm{1}(y\leqslant 0)]_\mathsf{B}$. Finally, the submission $x=\Delta\theta_t^i$ will be accepted if $\mathsf{Valid}(x)=1$, or be rejected otherwise.

Note that in practice, the fix-point representation might incur a very small difference between the original real value $x^\top x$ and the result computed in the field $\mathbb{F}$, which could make the protocol mistakenly reject a valid submission. To address this problem, we can instead use a slightly larger $C$ (such as $C\leftarrow C+10^{-5}$) in secure validation, while the clients are still required to clip their updates with the original value of $C$.

\subsection{Security Analysis}

We discuss the correctness and security (soundness and zero-knowledge) properties of secure aggregation and secure validation in PRECAD as follows. 

1) Correctness. If client $\mathsf{C}_i$ is honest (on clipping $\Delta\theta_t^i$ with bound $C$), the servers will always accept $\Delta\theta_t^i$. The correctness of the scheme follows by construction.

2) Soundness. In our scenario, soundness means a malicious client, who does not clip the norm of update with bound $C$, will be detected in secure validation (except negligible probability due to the fix-point representation). Specifically, server $\mathsf{S_A}$ and $\mathsf{S_B}$ check each client' submission $\Delta\theta_t^i$ and assign the result to $\mathsf{Valid}(\Delta\theta_t^i)$, which concludes that the servers either accept $\Delta\theta_t^i$ (when $\mathsf{Valid}(\Delta\theta_t^i)=1$) or reject $\Delta\theta_t^i$ (when $\mathsf{Valid}(\Delta\theta_t^i)=0$). Thus, any malicious client must either submit a well-formed submission or be treated as invalid.

3) Zero-Knowledge. For \emph{secure aggregation}, the random splitting of secret sharing guarantees that each server gains no information about a client's submission except the result of aggregation, where the information leakage from the aggregation result is bounded and quantified via DP (refer to Sec. \ref{sec:privacy_analysis}). For \emph{secure validation}, the two servers implement Beaver's multiplication, where Beaver's analysis \cite{beaver1991efficient} guarantees that the multiplication gate leaks no information to the servers. Note that the secure addition and secure comparison are executed by each server locally without communicating to each other, thus these steps leak no information as well. Therefore, the secure validation leaks no information except the validation result, which is always valid for benign clients, thus even the validation result leaks no information of \emph{benign} clients. Finally, the whole protocol is unconditionally secure  because each cryptographic primitive is unconditionally secure against adversaries with unbounded computation power, and they can be securely composed \cite{goldreich2009foundations}.

\section{Privacy and Robustness Analysis}
\label{sec:privacy_robustness_analysis}

Our privacy-preserving FL framework in Sec. \ref{sec:framework} \emph{simulates} CDP under the assumption that the two servers are honest-but-curious and non-colluding. Thus, PRECAD achieves SIM-CDP \cite{mironov2009computational} in FL. In the rest of this paper, we use DP to indicate SIM-CDP for the ease of presentation. Recall that PRECAD utilizes two-levels (both record-level and client-level) of clipping, where record-level clipping is guaranteed for benign clients, and client-level clipping is guaranteed for all clients (including malicious clients) because the servers perform secure validation. Therefore, it provides record-level DP for \emph{benign clients} and client-level DP for \emph{all clients}, where the former is our privacy goal, and the latter can be shown to provide robustness against malicious clients who implement model poisoning attacks.

\subsection{Record-level Differential Privacy}
\label{sec:privacy_analysis}

Without loss of generality, we assume the attacker wants to infer one record of a benign victim client $\mathsf{C}_i$. Recall that the attacker corrupts at most one server and any number of clients (except the victim client $\mathsf{C}_i$), where all corrupted parties are assumed as honest-but-curious (a.k.a. semi-honest).  Note that PRECAD provides record-level privacy for all \emph{benign} clients who implement the protocol honestly, while the privacy of \emph{malicious} clients (who may deviate from the protocol) is not guaranteed. Since corrupting one or neither of the two servers have different privacy guarantees, our results have two cases, which are shown in the following theorem (the involved notations are summarized in Table \ref{tab:notations}).

\begin{theorem}[Privacy Analysis]
\label{thm:privacy_analysis}
Denote the set of parties corrupted by the attacker (for privacy breach purpose) as $\mathcal{A}$. Then, Algorithm \ref{alg:framework} satisfies record-level $(\epsilon,\delta_i)$-DP for a benign client $\mathsf{C}_i$ (i.e., $\mathsf{C}_i\notin\mathcal{A}$)  with any $\epsilon\geqslant0$ and 
\begin{align}
    \label{equ:delta_i}
	\delta_i(\epsilon)=\Phi\left(-\frac{\epsilon}{\mu_i}+\frac{\mu_i}{2}\right)-e^\epsilon\cdot\Phi\left(-\frac{\epsilon}{\mu_i}-\frac{\mu_i}{2}\right),
\end{align}
where $\Phi(\cdot)$ denotes the cumulative distribution function (CDF) of standard normal distribution, and $\mu_i$ is defined by 
\begin{align}
    \label{equ:mu_i}
    \mu_i=
    \begin{cases}
    p_i\sqrt{T_i(e^{1/\sigma^2}-1)}, & \text{ if } \mathsf{S_A}\in\mathcal{A} \text{ or } \mathsf{S_B}\in\mathcal{A} \\
    qp_i\sqrt{T(e^{1/(2\sigma^2)}-1)}, & \text{ if } \mathsf{S_A}, \mathsf{S_B}\notin\mathcal{A}
    \end{cases}
\end{align}
\end{theorem}
\begin{proof}
    See Appendix \ref{apx:proof_thm_privacy_analysis}.
\end{proof}

In Theorem \ref{thm:privacy_analysis}, the number of corrupted semi-honest clients (except the victim client $\mathsf{C}_i$) does not affect the privacy analysis because the privacy leakage of all iterations are accounted (which corresponds to the worst case where the leaked information in all iteration are catched by the attacker). The parameter $\mu_i$ quantifies the indistinguishability of each record in-or-out of client $\mathsf{C}_i$'s dataset, where a smaller $\mu_i$ indicates a higher indistinguishability, thus provides stronger privacy guarantees (similar to the role of privacy budget). In \eqref{equ:delta_i}, $\delta_i(\epsilon)$ is a decreasing function w.r.t. the privacy budget $\epsilon$. It reflects the trade-off between the privacy budget $\epsilon$ and the small probability $\delta$ in $(\epsilon,\delta)$-DP.  Also, a smaller $\mu_i$ incurs a smaller value of $\delta$. Therefore, if we fix the value of function $\delta_i(\epsilon)$, the value of $\epsilon$ will be decreased with a decreasing $\mu_i$, which yields a stronger DP guarantee.

The two cases in \eqref{equ:mu_i}, i.e., whether the attacker corrupts one server or not, mainly differ in 1) the influence of the probability $q$ that each client is selected by the servers; and 2) the multiplier of the noise ($\sigma^2$ or $2\sigma^2$) that affect the quantification of DP guarantees. We explain the reasons below:

1) Since each client is randomly selected with probability $q$ in each iteration, and there are $T$ global iterations, the expected number of iterations that the victim $\mathsf{C}_i$ participates in is $qT$ (i.e., $\mathbb{E}[T_i]=qT$). Then, the factor $\sqrt{T_i}$ of $\mu_i$ in the first case can be approximated as $\sqrt{qT}$ (v.s. $q\sqrt{T}$ in the second case), where $\sqrt{qT}$ provides weaker privacy protection than $q\sqrt{T}$ because $\sqrt{qT}>q\sqrt{T}$ (note that $q\in[0,1]$). The intuition is that the privacy amplification from client-level sampling holds when the sampling result is a \emph{random variable} in the attacker's view (thus improves the randomness and privacy). However, in the first case, the attacker corrupts one server and can identify the iterations that the victim $\mathsf{C}_i$ actually participates in. Thus, the client-level sampling in this case only reduces the participating times, where only participated iterations leak $\mathsf{C}_i$'s privacy. In contrast, the attacker in the second case is uncertain about the participating iterations, then $\mathsf{C}_i$ enjoys more privacy benefit from the randomness caused by sampling probability $q$.  We note that although the random sampling/selection could be done via cryptography primitives, one of the servers ultimately will know the selection results.

2) Recall that both two servers add Gaussian noise with variance $(R\sigma)^2$ in aggregation. In the first case, no matter which server is corrupted by the attacker, only one of the Gaussian noise provides DP because the corrupted server can cancel out its noise from the aggregation. However, in the second case, both two Gaussian noises are valid against the attacker, thus the variance of the noise is $2(R\sigma)^2$ in the attacker's view. Note that the parameter $R$ in  $\mathcal{N}(0,(R\sigma)^2\cdot\mathbf{I})$ is finally cancelled out in \eqref{equ:mu_i} because the record-level clipping ensures the record-level sensitivity to be $R$. 

\begin{table}[t!]
	\footnotesize
	\centering
	\caption{Notations in Privacy/Robustness Analysis}
	\vspace{-3mm}
	\begin{tabular}{c|l}
	\hline
	& \multicolumn{1}{c}{Definition}\\
	\hline
	$\mathsf{S_A}, \mathsf{S_B}$ & Two non-colluding servers\\
	$\mathsf{C}_i$ & The $i$-th client (who is the victim of privacy breach)\\
	$T$ & The total number of global iterations  \\
	$T_i$ & The total number of iterations that client $\mathsf{C}_i$ participates in\\
	$q$ & The probability of each client being selected in each iteration\\
	$p_i$ & The probability of each record being sampled by client $\mathsf{C}_i$ \\
	$R$ & The maximum gradient norm  for each \emph{record} \\ 
	$C$ & The maximum local model update norm  for each \emph{client} \\
	$\sigma$ & The multiplier of the additive Gaussian noise $\mathcal{N}(0,(R\sigma)^2\cdot\mathbf{I})$\\
	$\Phi(\cdot)$ & The  CDF of standard normal distribution \\
	\hline
	\end{tabular}
	\label{tab:notations}
\end{table}

\subsection{Robustness against Poisoning Attacks}
\label{sec:robustness_analysis}

Recall that for model poisoning attacks, we assume the servers are trusted, and the attacker corrupts a set of \emph{malicious} clients $\mathcal{A}_K$ with group size $K$. The set of all \emph{benign} clients is denoted as $\mathcal{C}$,  and the set of all participating clients (including both benign and malicious ones) is denoted as $\mathcal{C}^\prime_K\coloneqq\mathcal{C}\cup\mathcal{A}_K$. Consider the randomized learning mechanism $\mathcal{M}$ (i.e., Algorithm \ref{alg:framework}) with the input to be the training dataset of $\mathcal{C}$ or $\mathcal{C}^\prime_K$, then  $\mathcal{M}(\mathcal{C}^\prime_K)$ and $\mathcal{M}(\mathcal{C})$ are two distributions of the final model parameter $\theta$ learned from the dataset with or without the participation of malicious clients $\mathcal{A}_K$. The robustness of PRECAD against such attacker focuses on a fixed record $z$ in the testing phase and a bounded loss function $\ell(\theta,z)\in[0,B]$ for any $\theta$ in the model parameter space. We denote the \emph{expected loss} over the random model parameter $\theta$ on distributions  $\mathcal{M}(\mathcal{C})$ and $\mathcal{M}(\mathcal{C}^\prime_K)$ as
\begin{align}
    \label{equ:expected_loss}
    \mathcal{L}=\mathbb{E}_{\theta\sim\mathcal{M}(\mathcal{C})}[\ell(\theta,z)],\quad
    \mathcal{L}^\prime_K=\mathbb{E}_{\theta\sim\mathcal{M}(\mathcal{C}^\prime_K)}[\ell(\theta,z)]
\end{align}
where $\mathcal{L}$ is the expected loss without attack, and $\mathcal{L}^\prime_K$ is the expected loss under the attack with $K$ malicious clients. The following theorem states that due to the participation of malicious clients $\mathcal{A}_K$, the attacked loss $\mathcal{L}^\prime_K$ would not be very far away from the unattacked loss $\mathcal{L}$ (refer to Table \ref{tab:notations} for the definitions of parameter notations).

\begin{theorem}[Robustness against model poisoning attacks] 
\label{thm:robustness_analysis}
	For the randomized mechanism $\mathcal{M}$ in Algorithm \ref{alg:framework},  the expected loss $\mathcal{L}^\prime_K$ defined in \eqref{equ:expected_loss} on the model $\theta$ with poisoning attack has the following upper-bound and lower-bound (w.r.t. the expected loss $\mathcal{L}$ on unattacked model):
	\begin{align}
    \label{equ:robust_upper_bound}
	\mathcal{L}^\prime_K&\leqslant~ \inf_{\epsilon\geqslant0}~e^\epsilon\cdot\mathcal{L}+ B\cdot\delta(\epsilon)  \\
	\label{equ:robust_lower_bound}
	\mathcal{L}^\prime_K&\geqslant~ \sup_{\epsilon\geqslant0}~e^{-\epsilon}\cdot(\mathcal{L}- B\cdot\delta(\epsilon) )
    \end{align}
    Recall that $[0,B]$ is the range of the loss function $\ell(\theta)$. The function $\delta(\epsilon)$ is defined by
	\begin{align}
	    \label{equ:delta_robustness}
	    \delta(\epsilon)=\Phi\left(-\frac{\epsilon}{K\mu}+\frac{K\mu}{2}\right)-e^\epsilon\cdot\Phi\left(-\frac{\epsilon}{K\mu}-\frac{K\mu}{2}\right),
	\end{align}
	where $\Phi(\cdot)$ is the CDF of standard normal distribution, and $\mu$ is computed by
	\begin{align}
	    \label{equ:mu_robustness}
	    \mu=q\sqrt{T(e^{1/\tilde{\sigma}^2}-1)}, \text{ where }
	    \tilde{\sigma}=\sqrt{2}\sigma R/C
	\end{align}
\end{theorem}
\begin{proof}
    (Sketch) We first show that releasing the final model $\theta$ satisfies client-level $(\epsilon,\delta(\epsilon))$-DP for any $\epsilon\geqslant0$. Then, by leveraging the property of DP, we can bound the expected loss on the poisoned model. The factor $K$ is introduced in the group privacy of DP. Refer to Appendix \ref{apx:proof_thm_robustness_analysis} for the full proof.
\end{proof}

The robustness analysis in Theorem \ref{thm:robustness_analysis} does not necessarily depend on the details of the attack implementation, such as what  auxiliary information the attacker has, how the local poisoning model is trained, and the value of scaling factor to make a larger impact of the poisoning attack. The only assumption on the attacker is the value of $K$ (i.e., the number of malicious clients). We can observe that the robustness guarantee is stronger when $B$ and $K\mu$ are small, where a small $B$ makes the range of loss small, and a small $K$ limits the controlling capability of the attacker. From \eqref{equ:mu_robustness}, $\mu$ would be small when: $\sigma R$ is large, $C$ is small, and $q\sqrt{T}$ is small, where a larger $\sigma R$ introduces more noise, and smaller $C,q,T$ restrict the manipulation capability of the attacker. Note that such changes might also impact the value of $\mathcal{L}$ and usually the main task becomes less accurate (since more noise is introduced or the global model learns less from benign clients' data), which reflects the trade-off between robustness and utility.

The function $\delta(\epsilon)$ in \eqref{equ:delta_robustness} has a similar form as in \eqref{equ:delta_i}, but differs in the value of $\mu$ and introduces an additional factor $K$. Compared with \eqref{equ:mu_i}, the constant $\mu$ in \eqref{equ:mu_robustness} no longer depends on the record-level sampling probability $p_i$ because we quantify the bound on each client's contribution (where $p_i$ only influences the contribution of one record). Note that malicious clients can opt out \emph{record-level} clipping, but \emph{client-level} clipping can be securely verified by the servers without leaking any additional information of benign clients. The ratio $R/C$ occurs in \eqref{equ:mu_robustness} due to the different sensitivity on record-level and client-level. The factor $K$ is introduced in \eqref{equ:delta_robustness} via group privacy of DP because the neighboring datasets under the poisoning attack differ in a group of (malicious) clients with size $K$. A larger $K$ leads to a larger $\delta(\epsilon)$, i.e., the privacy guarantee drops with the size of the group, because the distance of neighboring datasets in group privacy is $K$ (as versus distance of 1 in the original DP).

\section{Evaluation}

In this section, we demonstrate the enhancement of PRECAD on privacy-utility tradeoff and poisoning robustness via experimental results on MNIST \cite{lecun1998mnist} and CIFAR-10 \cite{krizhevsky2009learning} datasets. All experiments are developed in Python. The experimental settings of FL mainly follow the previous work \cite{zheng2021federated}, and the cryptographic protocols are implemented under CrypTen library \cite{knott2021crypten}.

\subsection{Experimental Setup}
\label{sec:experiment_setup}

\textbf{Baselines.} We use  1) non-private and 2) LDP-based solution in FL as the baseline approaches. In non-private setting, all clients neither implement clipping nor noise augmentation, and the server just aggregates the submissions and then updates the global model. In LDP setting, benign clients implement record-level clipping (with norm clipping bound $R$) and Gaussian noise augmentation (with standard deviation $\sigma R$) on the summation of multiple record gradients, then send the noisy update to the server in plaintext. Note that the privacy accountant (i.e., the value of privacy budget $\epsilon$) of LDP is the same as the privacy analysis of PRECAD in the threat model with one corrupted server (i.e., the first case in Theorem \ref{thm:privacy_analysis}). We do not include CDP as the baseline because CDP assumes a trusted server, and our scheme PRECAD essentially is a simulation of CDP. And DDP is not included because DDP will reduce to LDP in the worst case (i.e., when $\tau=1$ in Table \ref{tab:DP_compare}).

\textbf{Datasets (non-IID).} We use two datasets for our experiments: MNIST \cite{lecun1998mnist} and CIFAR-10 \cite{krizhevsky2009learning}, where the default value of the number of total clients is $n=100$.  To simulate the heterogeneous data distributions, we make non-i.i.d. partitions of the datasets, which is a  similar setup as \cite{zheng2021federated} and is described below:

1) Non-IID MNIST: The MNIST dataset contains 60,000 training images and 10,000 testing images of 10 classes. There are 100 clients, each holds 600 training images. We sort the training data by digit label and evenly divide it into 400 shards. Each client is assigned four random shards of the data, so that most of the clients have examples of three or four digits.

2) Non-IID CIFAR-10: The CIFAR-10 dataset contains 50,000 training images
and 10,000 test images of 10 classes. There are 100 clients, each holds 500 training images. We sample the training images for each client using a Dirichlet distribution with hyperparameter 0.5.

\textbf{Evaluation Metrics.}
We consider \emph{main task accuracy} and \emph{backdoor accuracy} (if applicable) as the evaluation metrics. The former is measured on the original test dataset (without backdoor images), while the latter is measured on a modified version of the test dataset, where \emph{a lower backdoor accuracy indicates a stronger robustness} against backdoor attacks. The detailed implementation of backdoor attacks and definition of backdoor accuracy are described in Appendix \ref{apx:experimental_details}.

\textbf{Model Architecture.}
For MNIST dataset, we use the CNN model from PyTorch example\footnote{https://github.com/pytorch/opacus}. For CIFAR-10 dataset, we use the CNN model from the TensorFlow tutorial\footnote{https://www.tensorflow.org/tutorials/images/cnn}, like the previous works \cite{zheng2021federated,mcmahan2018learning}. The hyperparameters for training are described in Appendix \ref{apx:experimental_details}.

\subsection{Privacy-Utility Tradeoff}

\begin{figure}[!t]
    \centering
    \includegraphics[width=1.6in]{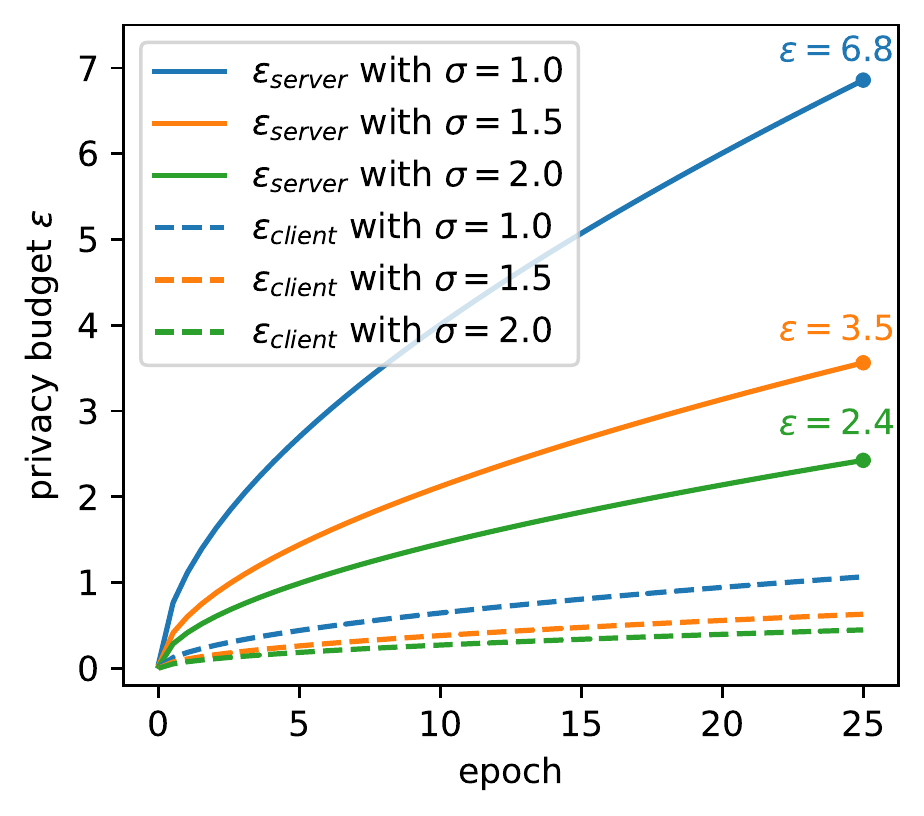}
    \includegraphics[width=1.7in]{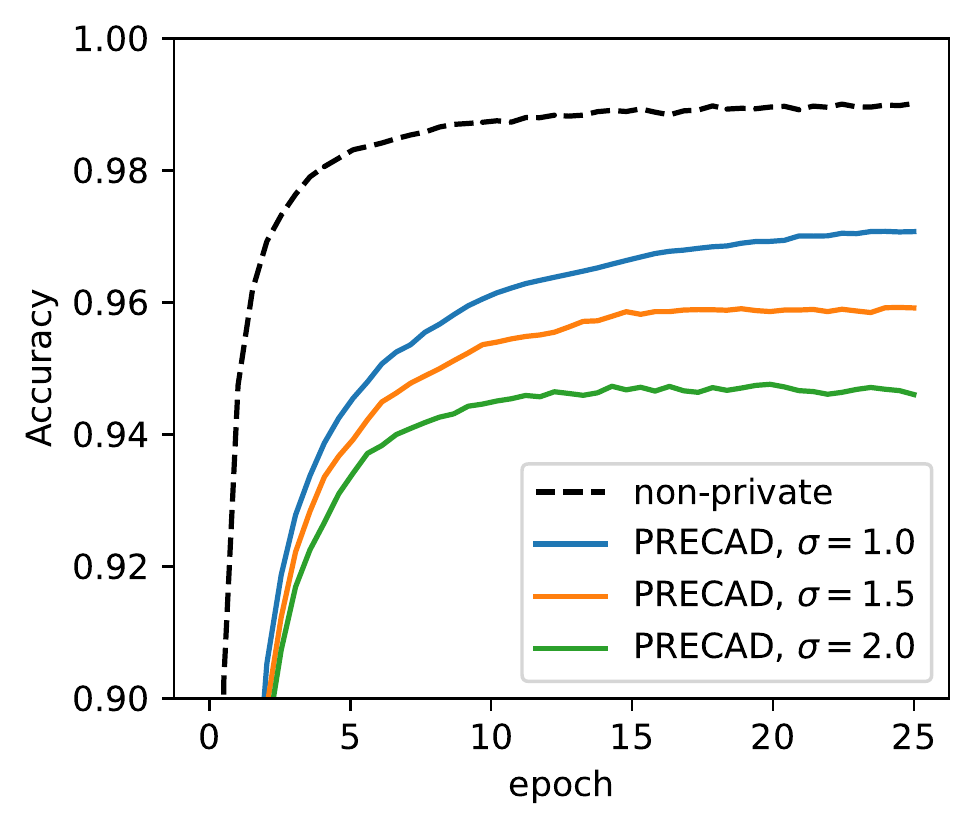}
    \vspace{-6mm}
    \caption{Privacy budget curve (left) and accuracy curve (right) of PRECAD w.r.t. the epoch, where one epoch equals to 200 global iterations because we set the sampling ratios as $p_i=0.05$ (for record) and $q=0.1$ (for client), thus $200\cdot p_i\cdot q=1$. The value of  $\epsilon_{\text{server}}$ (solid lines) and $\epsilon_{\text{client}}$ (dashed lines) represent the privacy budget $\epsilon$ (under fixed $\delta=10^{-5}$) in the two cases where one of the two servers is corrupted by the attacker or not. In the first case, PRECAD provides less privacy guarantee (thus a larger value of $\epsilon_{\text{server}}$) than in the second case with privacy budget $\epsilon_{\text{client}}$. } 
    \label{fig:privacy_curve}
\end{figure}

\begin{figure}[!t]
    \centering
    \includegraphics[width=1.65in]{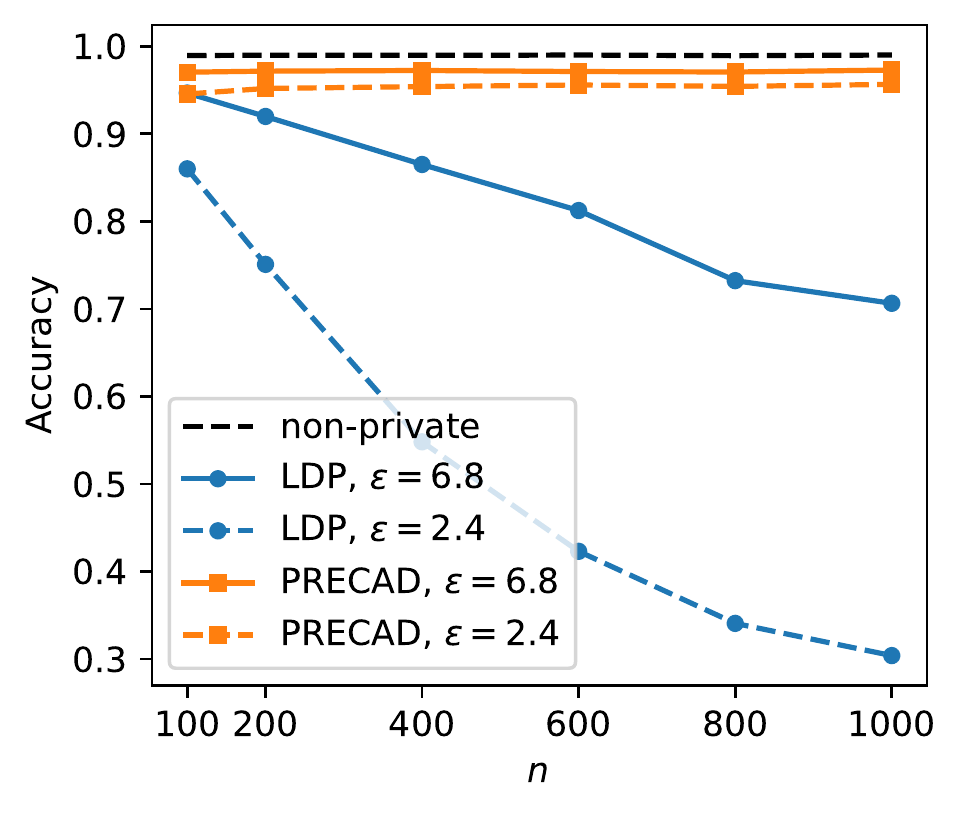}
    \includegraphics[width=1.65in]{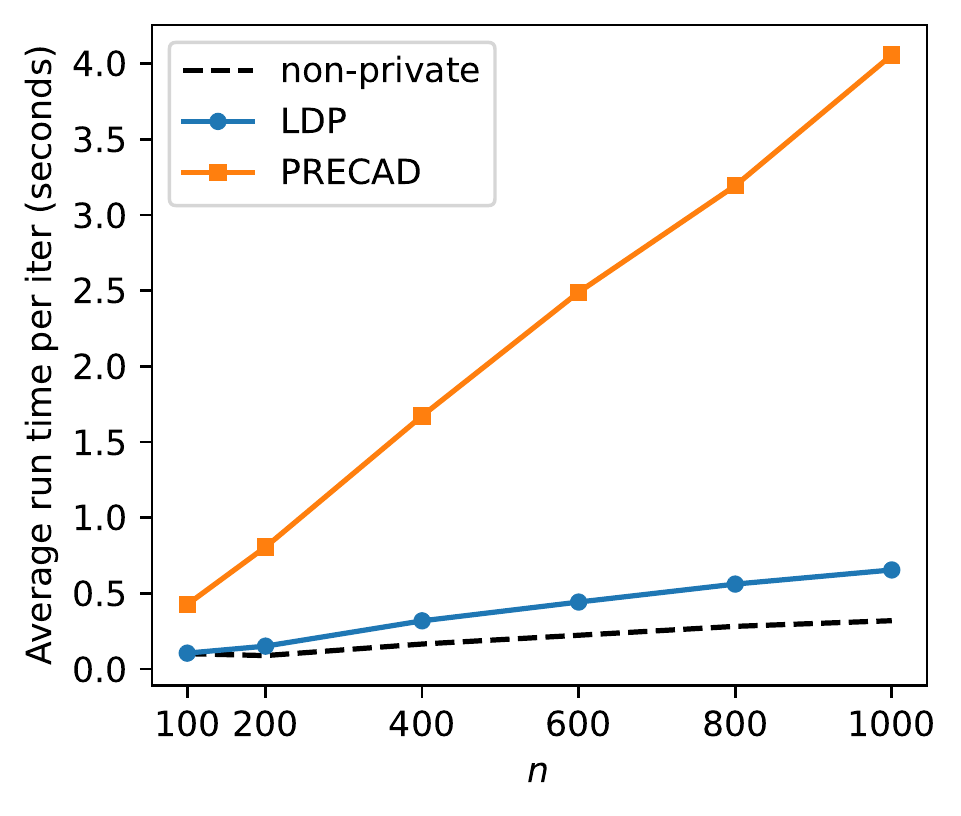}
    \vspace{-6mm}
    \caption*{(a) MNIST Dataset}
    \vspace{2mm}
    \includegraphics[width=1.65in]{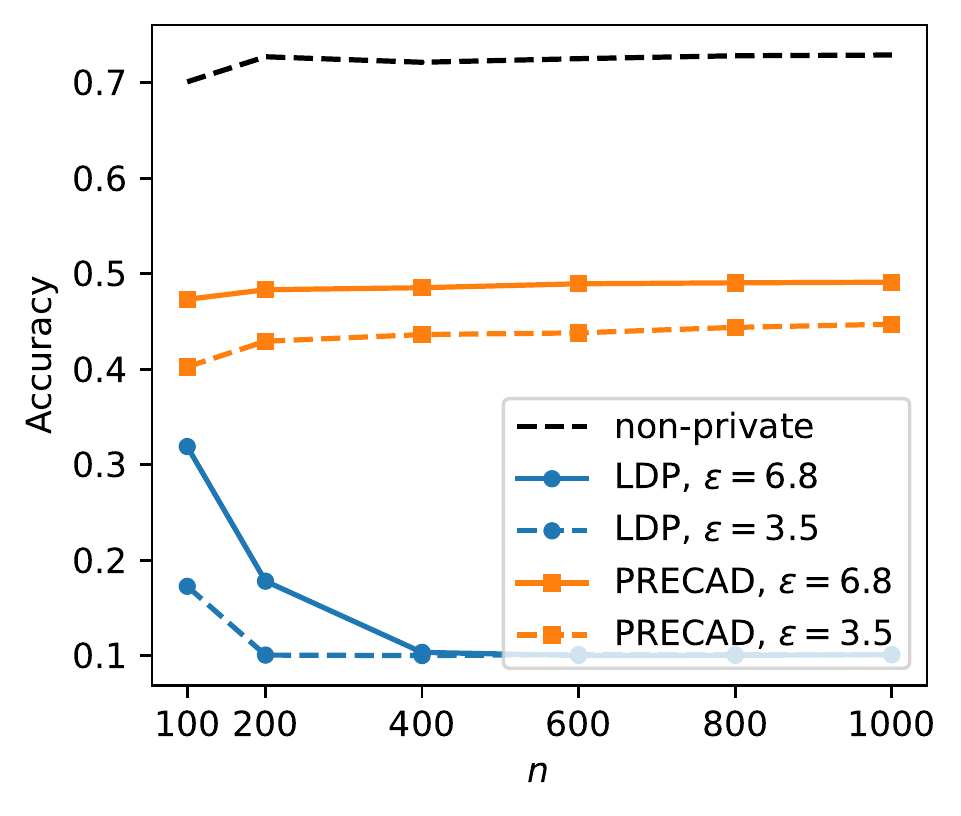}
    \includegraphics[width=1.65in]{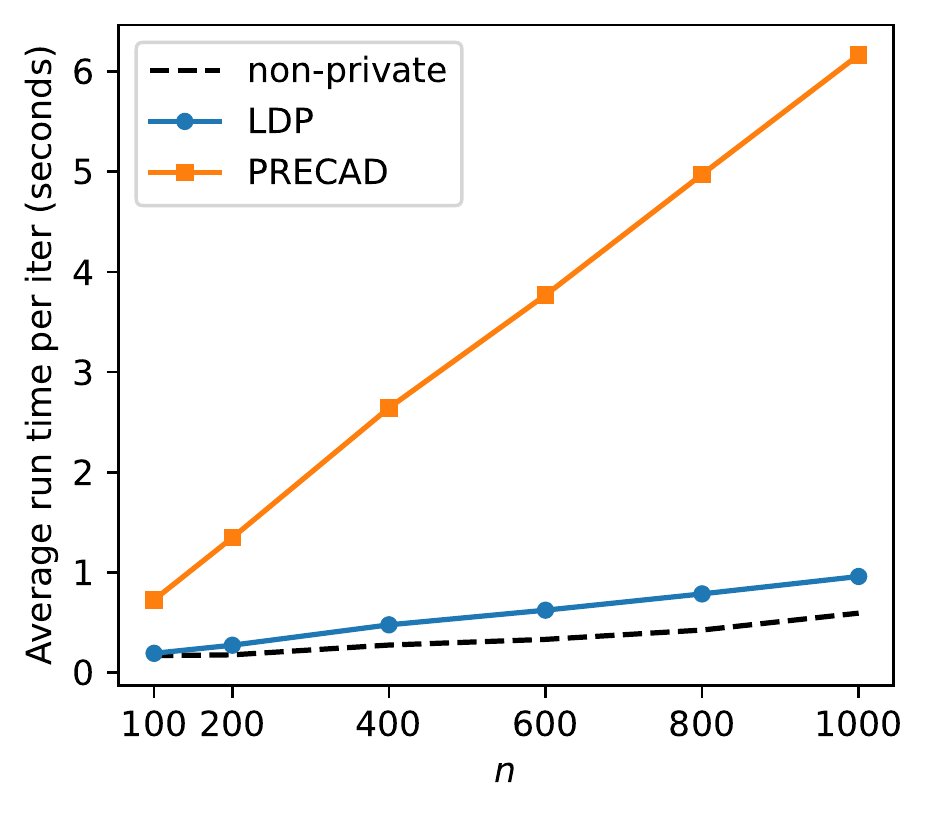}
    \vspace{-6mm}
    \caption*{(b) CIFAR-10 Dataset}
    \caption{Comparison of accuracy (left) and efficiency (right) under different $n$ (total number of clients), where efficiency is quantified by the average running time per iteration. Each client participates in each iteration with probability $q=0.1$.}
    \vspace{-2mm}
    \label{fig:vary_n}
\end{figure}

In this subsection, we show the privacy-utility tradeoff when there are no backdoor attacks. In this scenario, we disable the \emph{client-level} clipping in the client-side (Line-7 in Algorithm \ref{alg:local_updates}) and secure validation in the server-side (Line-4 in Algorithm \ref{alg:framework}) of PRECAD for a fair comparison. Note that the record-level clipping, additive secret sharing scheme, and Gaussian noise augmentation are always maintained to achieve the record-level DP as discussed in Sec. \ref{sec:privacy_analysis}.

\textbf{Privacy and Accuracy Curves of PRECAD.}
Figure \ref{fig:privacy_curve} shows the privacy cost  and the accuracy (for MNIST dataset) of PRECAD with respect to the epoch/iteration. In each fixed epoch, we can observe the privacy-utility tradeoff: a larger noise leads to more private guarantee (i.e., a smaller privacy budget $\epsilon$) but also leads to less accuracy. In Figure \ref{fig:privacy_curve}, the privacy budget against a stronger attacker (who corrupts one server and any number of clients) is always larger than against a weaker attacker (who only corrupts clients). In the following, we only focus on the worst-case privacy budget in the two cases, where $\epsilon=6.8, 3.5, 2.4$ correspond to the noise multiplier $\sigma=1.0,1.5,2.0$ when the global model is trained with 25 epochs (i.e., $T=5000$) under default hyperparameters $q=0.1,p_i=0.05~(i=1,\cdots,n)$. Note that the privacy curve (i.e., privacy budget v.s. epochs) for CIFAR-10 is the same as MNIST because the hyperparameters are the same.

\textbf{Accuracy Comparison.} Figure \ref{fig:vary_n} (left) shows how the total number of clients (i.e., the value of $n$) affects the accuracy. We can observe that increasing the value of $n$ has no impact on the accuracy of non-private setting and PRECAD, but decreases the accuracy of LDP setting significantly, because more noise is aggregated when more clients participate in each iteration. Note that for CIFAR-10 dataset, the accuracy of non-private and PRECAD when $n=100$ is slightly lower than the cases when $n\geqslant200$. It is because the CIFAR-10 training examples of each client is sampled with random drawing (v.s. MNIST dataset is jointly divided), thus a small $n$ makes partial of training examples may not be included by any client.

\textbf{Efficiency Comparison.}
Figure \ref{fig:vary_n} (right) compares the computation  efficiency of different protocols, quantified by the average run time per iteration. With an increased $n$, the run time is increased for all approaches because more clients train their local models (we train local models in sequential due to the GPU memory limit). We can observe that PRECAD's run time is approximately 7$\times$  of the LDP-based solution, due to the additional computation introduced by cryptography, but the absolute time is acceptable (since the additive secret sharing scheme is a light-weight cryptographic technique). In addition, the non-private setting has less run time than the LDP setting, because the former does not need clipping and noise augmentation steps.

\subsection{Robustness against Backdoor Attacks}

\begin{figure}[!t]
    \centering
    \includegraphics[width=1.65in]{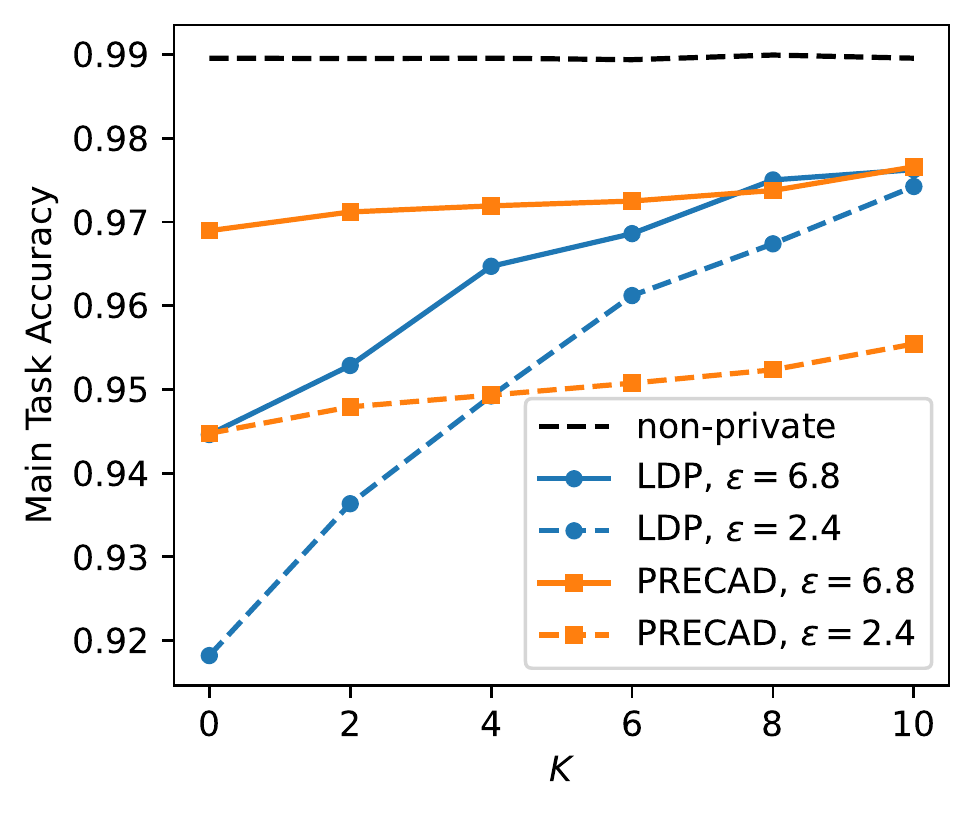}
    \includegraphics[width=1.65in]{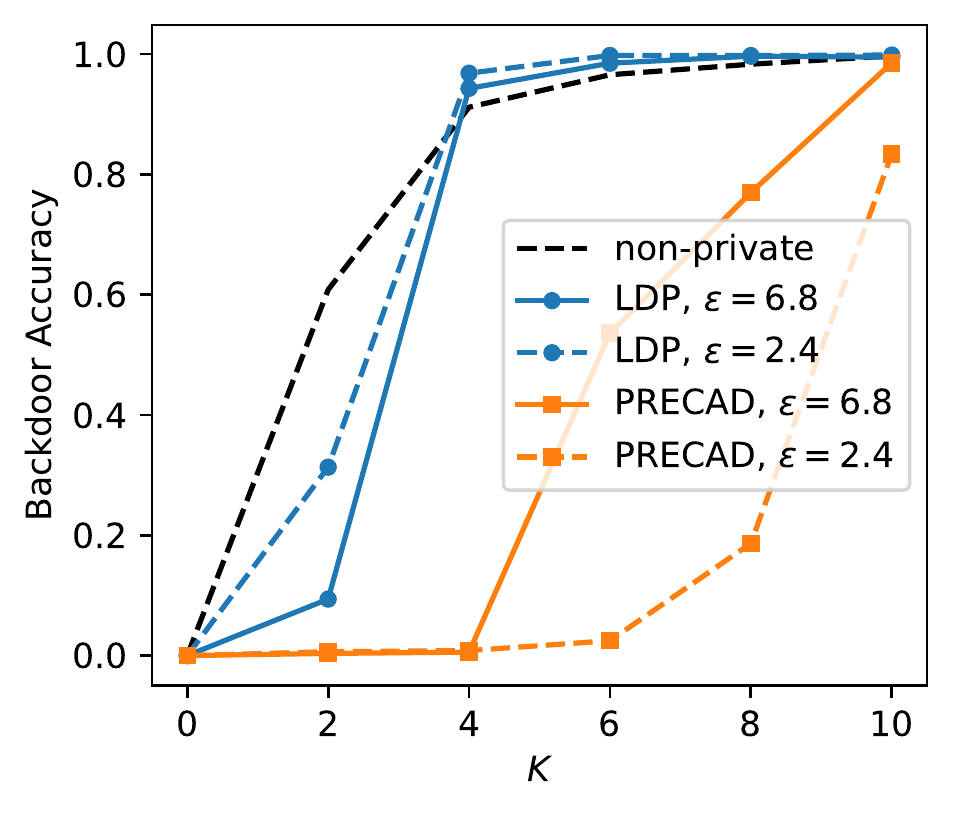}
    \vspace{-6mm}
    \caption*{(a) MNIST Dataset}
    \includegraphics[width=1.65in]{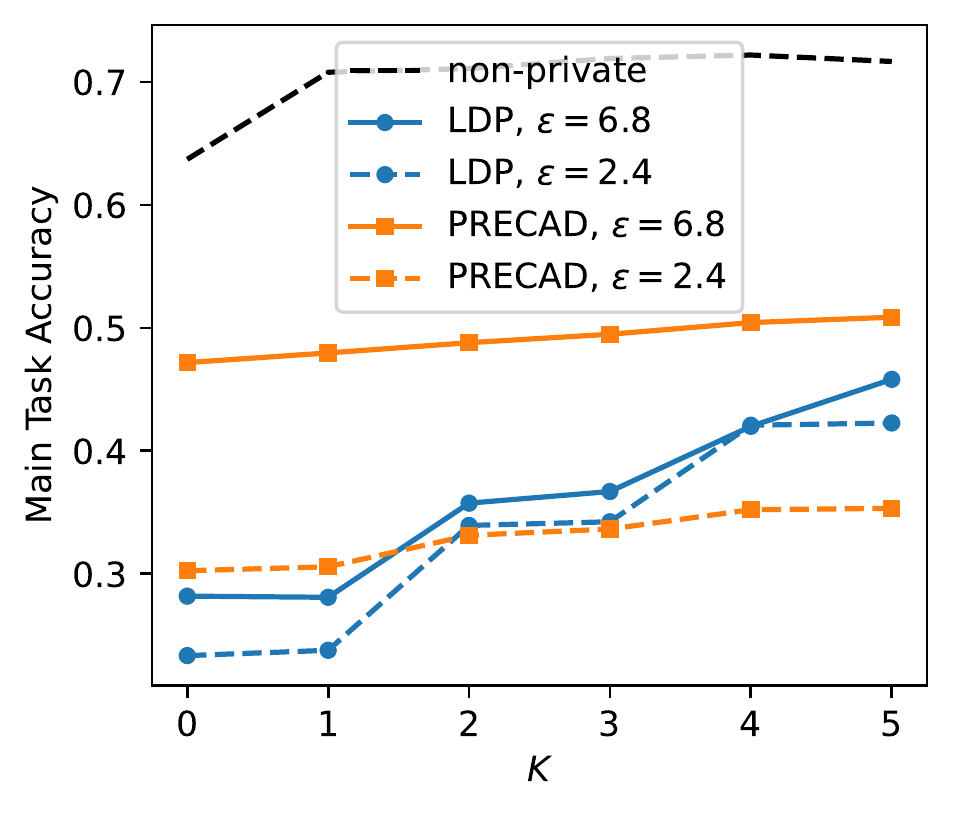}
    \includegraphics[width=1.65in]{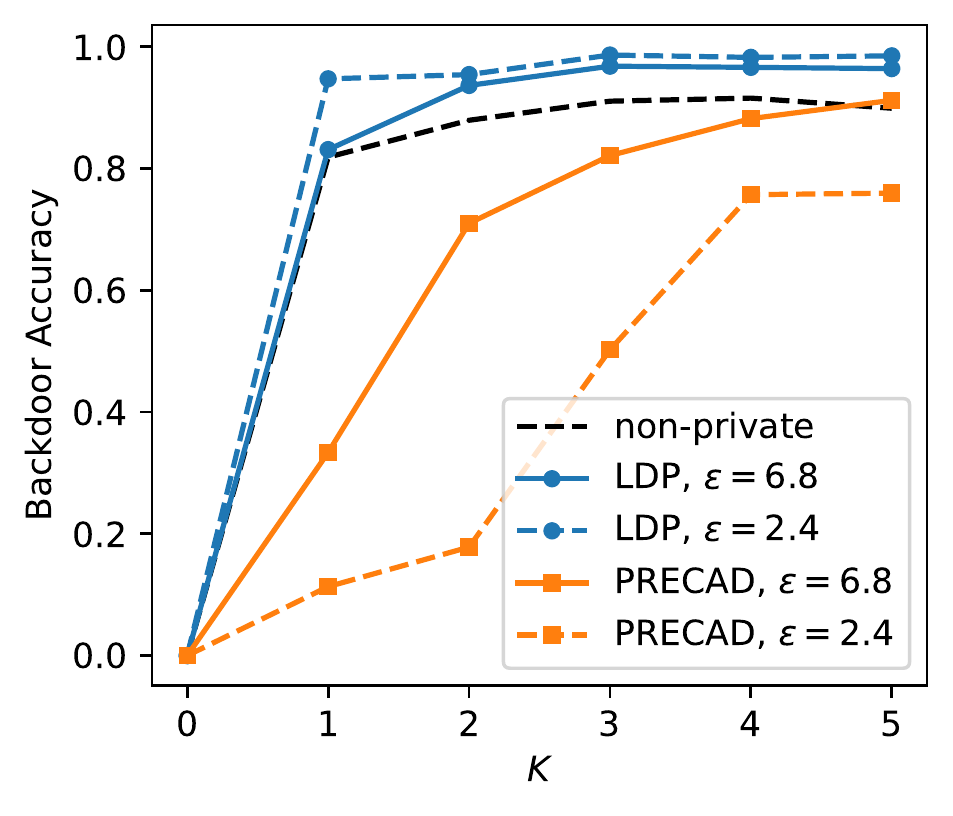}
    \vspace{-6mm}
    \caption*{(b) CIFAR-10 Dataset}
    \caption{Comparison of main task accuracy (left) and backdoor accuracy (right) under different $K$ (number of backdoor clients) under parameter $n=100$ and $C=30$.}
    \vspace{-2mm}
    \label{fig:vary_K}
\end{figure}

\begin{figure}[!t]
    \centering
    \includegraphics[width=1.65in]{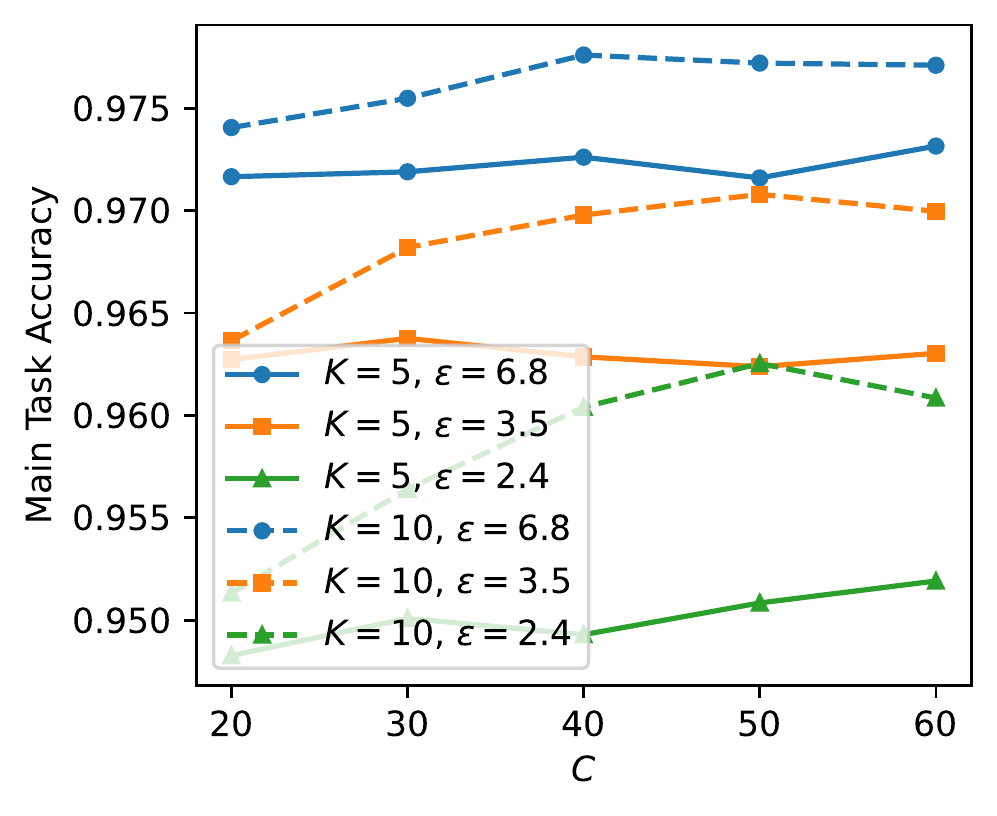}
    \includegraphics[width=1.65in]{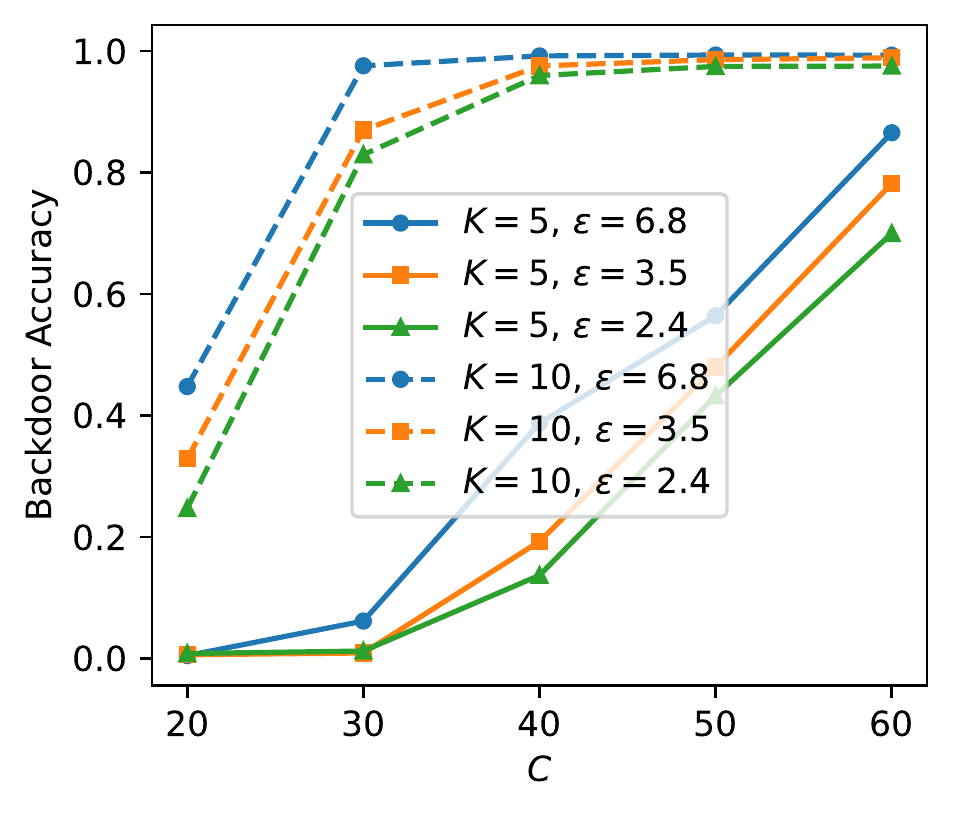}
    \vspace{-6mm}
    \caption*{(a) MNIST Dataset}
    \includegraphics[width=1.65in]{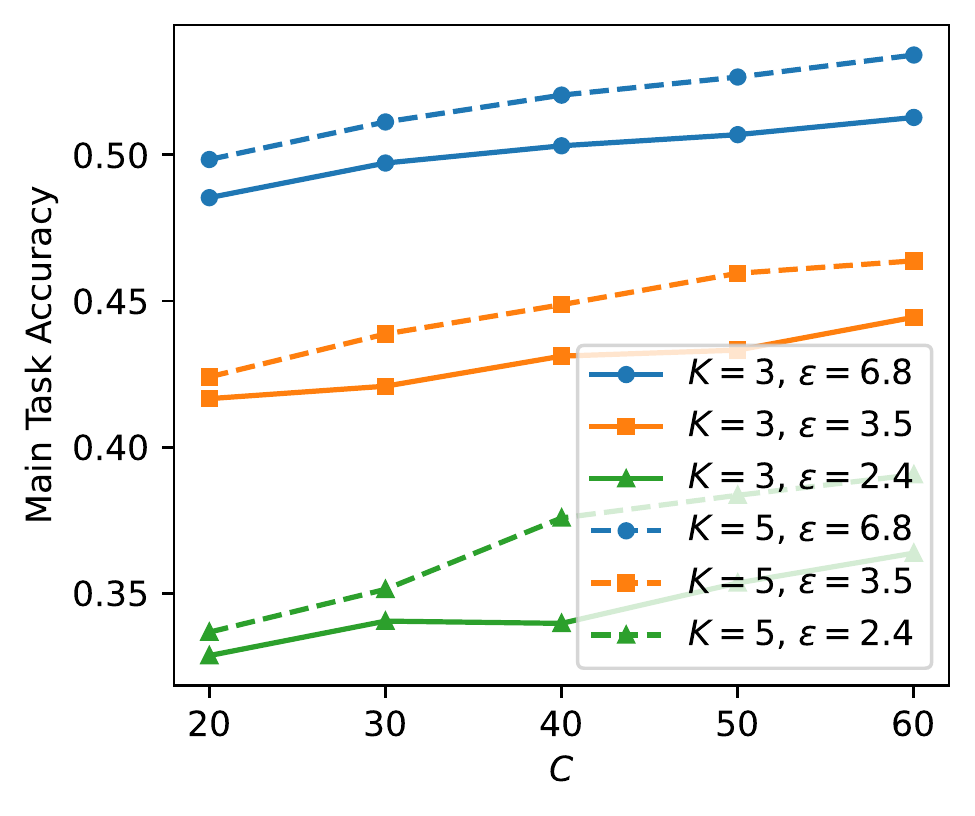}
    \includegraphics[width=1.65in]{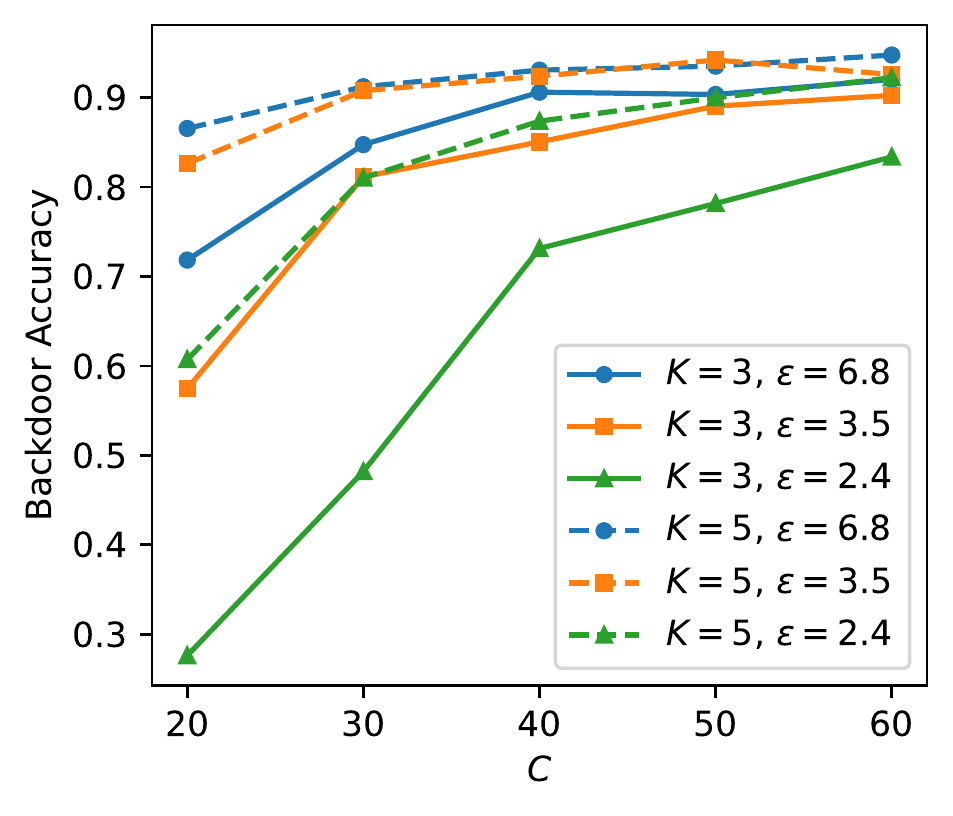}
    \vspace{-6mm}
    \caption*{(b) CIFAR-10 Dataset}
    \caption{The influence of $C$ (client-level clipping norm bound) on the main task accuracy (left) and backdoor accuracy (right) of PRECAD.}
    \vspace{-2mm}
    \label{fig:vary_C}
\end{figure}

In this subsection, we assume that the attacker corrupts $K$ (malicious) clients to implement backdoor attacks (refer to Appendix \ref{apx:experimental_details} for the attack details).  In PRECAD, all clients are required to implement client-level clipping with bound $C$; otherwise, the invalid submission will be identified, and won't be aggregated in the global model update. As a comparison, clients' submissions in LDP-based solution are also  clipped with bound $C$ (but after the noise being added by clients), which can be directly verified because clients upload plaintext submissions in LDP. However, for the non-private setting, we do not clip the submissions because benign clients do not implement record-level clipping in this setting.

\textbf{Influence of $K$.} 
Figure \ref{fig:vary_K} shows how the value of $K$  (i.e., the number of backdoor clients) affects the main task accuracy and backdoor accuracy under different levels of privacy guarantees (i.e., different amount of added Gaussian noise). In general, a large $K$ will increase the accuracy on both main task and backdoor task (recall that backdoor attacker's goal is to increase backdoor accuracy while maintaining main task accuracy), but the degree of increment might be different for different cases.  1) The non-private and LDP settings are vulnerable to backdoor attacks: merely 4 backdoor clients (for MNIST dataset) or 1 backdoor client (for CIFAR-10 dataset) can increase the backdoor accuracy from $0\%$ to more than $80\%$, where \emph{Semantic Backdoor Attack} in CIFAR-10 is stronger than \emph{Pixel-pattern Backdoor Attack} in MNIST. 2) PRECAD shows higher robustness against backdoor attacks than the other two, and smaller $\epsilon$ (i.e., with larger $\sigma$) yields lower backdoor accuracy, which indicates that noise in PRECAD enhances robustness. In contrast, the noise in LDP reduces robustness from the results. 3) More backdoor clients results in a slight improvement on main task accuracy for non-private setting and PRECAD, but a large improvement for LDP setting.  It is because the accuracy of the LDP setting is relatively low (since a larger amount of noise need to be added under the same privacy guarantee), and backdoor clients have additional advantages (they neither clip the gradient nor add noise) than benign clients on improving main task accuracy.

\textbf{Influence of $C$.}
Figure \ref{fig:vary_C} shows how the value of $C$ (i.e., the client-level clipping norm bound) affects the performance of PRECAD. We can observe that a smaller value of $C$ results in higher robustness against backdoor attacks, because of the more strictly bounded impact of the malicious clients. 1) For the MNIST dataset, when we set $C=20$, the backdoor accuracy is $0\%$ against $K=5$ malicious clients, and is less than $50\%$ against $K=10$ malicious clients, while the main task accuracy is reduced by only $1\%$, which is a negligible influence on the utility. 2) For the CIFAR-10 dataset under a stronger backdoor attack strategy, the backdoor accuracy is below $30\%$ when $K=3$ and $C=20$ under $\epsilon=2.4$, where the main task accuracy does not reduce too much. 3) On the impact of different $\epsilon$ for both datasets, a smaller $\epsilon$ (stronger privacy) leads to lower backdoor accuracy (stronger robustness) with very little impact on main task accuracy.

In summary, the noise for privacy purpose in PRECAD reduces utility but enhances robustness against backdoor attacks, which implies the robustness-utility tradeoff (similar observations have been made for adversarial robustness in \cite{lecuyer2019certified,cohen2019certified}). However, the noise in LDP setting reduces both utility and robustness (refer to Table \ref{tab:DP_compare} for a summary of the comparison of different approaches).

\section{Discussion}

\textbf{Limitations.} 
Though PRECAD improves both privacy-utility tradeoff of DP and poisoning robustness against malicious clients, it has several limitations: 1) The trust assumption of non-colluding servers is slightly strong, which might not hold for all application scenarios. 2) The utilized cryptography techniques, including secret sharing and MPC, incurs additional cost on both computation and communication. Thus, PRECAD might not be the most economic solution when the application has rigorous constraints on computation and communication.

\textbf{Generality.}
In PRECAD, the secret sharing scheme can be substituted by other cryptography primitives, such as the pairwise masking strategy in \cite{bonawitz2017practical}, with the condition that the client-level clipping can be securely verified by the server(s). Also, the perturbation mechanism can be substituted by other DP mechanisms, such as Laplace Mechanism \cite{dwork2006calibrating} and Exponential Mechanism \cite{mcsherry2007mechanism} according to the application scenario, but the method of privacy accountant and robustness analysis would be different.

\textbf{Guarantees under Malicious Setting.}
In PRECAD, we assume both servers implement the protocol honestly to guarantee the required privacy and robustness guarantees. However, if one of them, say server $\mathsf{S_A}$, maliciously deviates the protocol by omitting the noise augmentation (or adding less noise), PRECAD still provides the same privacy guarantees against $\mathsf{S_A}$ as in Theorem \ref{thm:privacy_analysis} because the noise providing DP against $\mathsf{S_A}$ is honestly added by $\mathsf{S_B}$. However, the privacy against corrupted clients and the robustness against poisoning attacks become weaker (than in Theorem \ref{thm:privacy_analysis} and Theorem \ref{thm:robustness_analysis}) since the overall noise added in the global model is reduced.

\textbf{A Faster Version with only Privacy Guarantee.}
If we only need to provide record-level privacy (i.e., without robustness requirement), then both the client-level clipping in the client-side and the secure validation in the server-side can be skipped. It would yield a more accurate model (since clipping introduces biased noise) and improve the computation efficiency.

\section{Related Work}

\subsection{Privacy-Preserving Federated Learning}
Existing approaches  on privacy-preserving federated learning  are typically designed based on cryptography and/or DP.

\textbf{Crypto-based.}
Aono et al. \cite{aono2017privacy} used additively Homomorphic Encryption to preserve the privacy of gradients and enhance the security of the distributed learning system. Mohassel et al. \cite{mohassel2017secureml}  proposed SecureML which conducts privacy-preserving learning via Secure Multi-Party Computation (MPC) \cite{yao1982protocols}, where data owners need to process, encrypt and/or secret-share their data among two non-colluding servers in the initial setup phase. Bonawitz et al. \cite{bonawitz2017practical}  proposed a secure, communication-efficient, and failure-robust protocol for secure aggregation of individual model updates.  However, all the above cryptography based protocols in some way prevent anyone from auditing clients' updates to the global model, which leaves spaces for the malicious clients to attack. For example, malicious clients can introduce stealthy backdoor functionality into the global model without being detected.

\textbf{DP-based.} Differential Privacy (DP)  was originally designed for the centralized scenario where a trusted database server, who has direct access to all clients' data in the clear, wishes to answer queries or publish statistics in a privacy-preserving manner by randomizing query results. In FL, McMahan et al. \cite{mcmahan2018learning} introduces two algorithms DP-FedSGD and DP-FedAvg, which provides client-level privacy with a trusted server. Geyer et al. \cite{geyer2017differentially} uses an algorithm similar to DP-FedSGD for the architecture search problem, and the privacy guarantee acts on client-level and  trusted server too. Li et al. \cite{li2020differentially} studies the online transfer learning and introduces a notion called task global privacy that works on record-level. However, the online setting assumes the client only interacts with the server once and does not extend to the federated setting. Zheng et al. \cite{zheng2021federated} introduced two privacy notions, that describe privacy guarantee against an individual malicious client and against a group of malicious clients on record-level privacy, based on a new privacy notion called $f$-differential privacy. However, the privacy analysis of this work does not consider the case when the server is corrupted, and the privacy budget in the worst-case adversary setting (i.e., all clients except the victim are malicious) is too large, thus does not provide meaningful privacy guarantee. 

\textbf{Hybrid Solutions.}
Truex et al. \cite{truex2019hybrid} proposed a hybrid solution  which utilizes threshold-based partially additive homomorphic encryption to reduce the needed noise for record-level DP guarantee.  Xu et al. \cite{xu2019hybridalpha} improved this hybrid solution with enhanced efficiency and accommodation of client drop-out. However, these hybrid solutions are vulnerable to malicious clients (due to encryption) and the utility gain from encryption is sensitive to the number of non-colluding parties.

\textbf{Privacy Amplification by Shuffling.} 
Different from using cryptography to improve privacy-utility tradeoff of DP, researchers introduced a shuffler model, which achieves a middle ground between CDP and LDP, in terms of both privacy and utility. Bittau et al. \cite{bittau2017prochlo} was the first to propose the shuffling idea, where a shuffler is inserted between the users and the server to break the linkage between the report and the user identification. Recent work by Cheu et al. \cite{cheu2019distributed} analyzed the differential privacy properties of the shuffler model and shows that in some cases shuffled protocols provide strictly better accuracy than local protocols.  Balle et al. \cite{balle2019privacy} provided a tighter and more general privacy amplification bound  result by leveraging a technique called blanket decomposition. We note that current shuffler models are mainly used in the application of local data aggregation.

\subsection{Robust Federated Learning}
FL systems are vulnerable to model
poisoning attacks, which aim to thwart the learning of the global model (a.k.a. Byzantine attacks) or hide a backdoor trigger into the global model (a.k.a. backdoor attacks). These attacks poison local model updates before uploading them to the server. More details of poisoning attacks and other threats of FL can be found from the survey paper \cite{lyu2020privacy}.

\textbf{Byzantine Robustness.}
Most state-of-the-art Byzantine-robust solutions play with mean or median statistics of gradient contributions. Blanchard et al. \cite{blanchard2017machine} proposed Krum which uses the Euclidean distance to determine which gradient contributions should be removed, and can theoretically withstand poisoning attacks of up to $33\%$ adversaries in the participant pool. Mhamdi et al. \cite{mhamdi2018hidden} proposed a meta-aggregation rule called Bulyan, a two-step meta-aggregation algorithm based on the Krum and trimmed median, which filters malicious updates followed by computing the trimmed median of the remaining updates.

\textbf{Backdoor Robustness.}
Andreina et al. \cite{andreina2020baffle} incorporates an additional validation phase to each round of FL to detect backdoor. Sun et al. \cite{sun2019can}  showed that clipping the norm of model updates and adding Gaussian noise can mitigate backdoor attacks that are based on the model replacement paradigm. Xie et al. \cite{xie2021crfl} provided the first general framework to train certifiably robust FL models against backdoors by exploiting clipping and smoothing on model parameters to control the global model smoothness. However, these works do not consider the privacy issue in FL.

\subsection{On both Privacy and Robustness}

Recently, some works tried to simultaneously achieve both privacy and robustness of FL. He et al. \cite{he2020secure} proposed a Byzantine-resilient and privacy-preserving solution, which makes distance-based robust aggregation rules (such as Krum \cite{blanchard2017machine}) compatible with secure aggregation via MPC and secrete sharing. So et al. \cite{so2020byzantine} developed a similar scheme based on the Krum aggregation, but rely on different cryptographic techniques, such as verifiable Shamir's secret sharing and Reed-Solomon code. Velicheti et al. \cite{velicheti2021secure} achieved both privacy and Byzantine robustness via incorporating secure averaging among randomly clustered clients before filtering malicious updates through robust aggregation. However, these works do not achieve DP against the server, since the aggregated model is directly revealed.

For the relationship between DP and robustness, Guerraoui et al. \cite{guerraoui2021differential} provided a theoretical analysis on the problem of combining DP and Byzantine resilience in FL frameworks. They concluded that the classical approaches to Byzantine-resilience and DP in distributed SGD (i.e., the LDP setting) are practically incompatible. Naseri et al. \cite{naseri2020toward} presented a comprehensive empirical evaluation to show that Local and Centralized DP (LDP/CDP) are able to defend against backdoor attacks in FL. However, the malicious clients in the LDP setting are assumed to follow the DP protocol honestly (which usually does not hold in practice), and the CDP setting assumes a trusted server (which is too strong).

\section{Conclusion}
In this paper, we developed a novel framework PRECAD for FL to enhance the privacy-utility tradeoff of DP and the robustness against model poisoning attacks, by leveraging secret sharing and MPC techniques. With the record-level clipping and securely verified client-level clipping, the noise added by servers provides both  record-level DP and client-level DP. The former is our privacy goal, and the latter is utilized to show the robustness against poisoning model updates that are uploaded by malicious clients. Our experimental results validate the improvement of PRECAD on both privacy and robustness. 

For future work, we will extend our framework to other cryptography primitives and DP mechanisms, and show certifiable robustness against model poisoning attacks.

\bibliographystyle{plain}
\bibliography{mybibfile}

\appendix

\section{Beaver's Multiplication Protocol}
\label{apx:Beaver_MPC}

In the context of additive secret sharing discussed in Sec. \ref{sec:preliminaries_secret_sharing}, we assume the $j$-th server holds shares $[x]_j$ and $[y]_j$ and wants to compute a share of $xy$. All arithmetic in this section is in a finite field $\mathbb{F}$. Beaver \cite{beaver1991efficient} showed that the servers can use pre-computed \emph{multiplication triples} to evaluate multiplication gates. A multiplication triple is a one-time-use triple of values $(a,b,c)\in\mathbb{F}^3$, chosen at random subject to the constraint that $a\cdot b=c\in\mathbb{F}$. When used in the context of multi-party computation, each server $j$ holds a share $([a]_j,[b]_j,[c]_j)\in\mathbb{F}^3$ of the triple. To jointly evaluate shares of the output of a multiplication gate $xy$, each server $j$ compute the following values:
\begin{align*}
    [d]_j=[x]_j-[a]_j,\quad
    [e]_j=[y]_j-[b]_j
\end{align*}
Each server $j$ then broadcasts $[d]_j$ and $[e]_j$. Using the broadcasted shares, every server can reconstruct $d$ and $e$, which allows each of them to compute
\begin{align*}
    z_j=de/s +d[b]_j+e[a]_j+[c]_j
\end{align*}
Recall that $s$ is the number of servers (which is a public constant) and the division symbol here indicates division (i.e., inversion then multiplication) in the field $\mathbb{F}$. A few lines of arithmetic confirm that $z_j$ is a sharing of the product $xy$:
\begin{align*}
    \sum\nolimits_j z_j&=\sum\nolimits_j (de/s +d[b]_j+e[a]_j+[c]_j) \\
    &=de+db+ea+c \\
    &=(x-a)(y-b) +(x-a)b +(y-b)a+c \\
    &=xy-ab+c=xy
\end{align*}

\section{Gaussian Differential Privacy (GDP)}
\label{apx:GDP}
\textbf{Privacy Accountant.} Since deep learning needs to iterate over the training data and apply gradient computation multiple times during the training process, each access to the training data incurs some privacy leakage from the overall privacy budget $\epsilon$. The total privacy leakage (or loss) of repeated applications of additive noise mechanisms follow from the composition theorems and their refinements \cite{dwork2014algorithmic}. The task of keeping track of the accumulated privacy loss in the course of execution of a composite mechanism, and enforcing the applicable privacy policy, can be performed by the privacy accountant. Abadi et al. \cite{abadi2016deep} proposed \emph{moments accountant} to provide a tighter bound on the privacy loss compared to the generic advanced composition theorem \cite{dwork2010boosting}. Another new and more state-of-the-art privacy accountant method is Gaussian Differential Privacy (GDP) \cite{dong2019gaussian,bu2020deep}, which was shown to obtain a tighter result than moments accountant.

\textbf{Gaussian Differential Privacy.}
GDP is a new privacy notion which faithfully retains  hypothesis testing interpretation of differential privacy. By leveraging the central limit theorem of Gaussian distribution, GDP has been shown to possess an \emph{analytically tractable} privacy accountant (vs. moments accountant must be done by numerical computation). Furthermore, GDP can be converted to a collection of $(\epsilon,\delta)$-DP guarantees (refer to Lemma \ref{lem:GDP_to_DP}). Note that even in terms of $(\epsilon,\delta)$-DP, the GDP approach gives a tighter privacy accountant than moments accountant.  GDP utilizes a single parameter $\mu\geqslant0$ (called privacy parameter) to quantify the privacy of a randomized mechanism. Similar to the privacy budget $\epsilon$ defined in DP,  a larger  $\mu$ in GDP indicates less privacy guarantee. Comparing with $(\epsilon,\delta)$-DP,  the new notion $\mu$-GDP can losslessly reason about common primitives associated with differential privacy, including composition, privacy amplification by subsampling, and group privacy. In the following, we briefly introduce some important properties (that will be used in the analysis of our approach) of GDP as below. The formal definition and more detailed results can be found in the original paper \cite{dong2019gaussian}. 

\begin{lemma}[Gaussian Mechanism for GDP \cite{dong2019gaussian}]
\label{lem:Gaussian_mechanism_GDP}
Consider the problem of privately releasing a univariate statistic $f(D)$ of a dataset $D$.  Define the sensitivity of $f(\cdot)$ as $s_f=\sup_{D,D^\prime}|f(D)-f(D^\prime)|$, where the supremum is over all neighboring datasets.
Then,  the Gaussian mechanism $\mathcal{M}(D)=f(D)+\xi$, where $\xi\sim\mathcal{N}(0,s_f^2/\mu^2)$, satisfies  $\mu$-GDP. 
\end{lemma}

\begin{lemma}[Composition Theorem of GDP \cite{dong2019gaussian}]
\label{lem:GDP_composition}
The $m$-fold composition of $\mu_i$-GDP mechanisms is $\sqrt{\mu_1^2+\cdots+\mu_m^2}$-GDP. 
\end{lemma}

\begin{lemma}[Group Privacy of GDP \cite{dong2019gaussian}]
	\label{lem:GDP_group}
	If a mechanism is $\mu$-GDP, then it is $K\mu$-GDP for a group with size $K$. 
\end{lemma}

\begin{lemma}[$\mu$-GDP to $(\epsilon,\delta)$-DP \cite{dong2019gaussian}]
	\label{lem:GDP_to_DP}
	A mechanism is $\mu$-GDP if and only if it is $(\epsilon,\delta(\epsilon))$-DP for all $\epsilon\geqslant0$, where
	\begin{align*}
		\delta(\epsilon)=\Phi\left(-\frac{\epsilon}{\mu}+\frac{\mu}{2}\right)-e^\epsilon\cdot\Phi\left(-\frac{\epsilon}{\mu}-\frac{\mu}{2}\right),
	\end{align*}
	and $\Phi$ denotes the CDF of standard normal (Gaussian) distribution. 
\end{lemma}

\begin{lemma}[Privacy Central Limit Theorem of GDP \cite{bu2020deep}]
	\label{lem:GDP_privacy_account}
	Denote $p$ as the subsampling probability, $T$ as the total number of iterations and $\sigma$ as the noise scale (i.e., the ratio between the standard deviation of Gaussian noise and the gradient norm bound). Then,
	 algorithm DP-SDG asymptotically satisfies $\mu$-GDP with privacy parameter $\mu=p\sqrt{T(e^{1/\sigma^2}-1)}$. 
\end{lemma}

In this paper, we use $\mu$-GDP as our primary privacy accountant method due to its good property on composition and accountant of privacy amplification in Lemma \ref{lem:GDP_privacy_account}, and then convert the result to $(\epsilon,\delta)$-DP via Lemma \ref{lem:GDP_to_DP}. We note that other privacy accountant methods, such as moments accountant \cite{abadi2016deep} and R{\'e}nyi DP (RDP) \cite{mironov2017renyi}, are also applicable to the proposed scheme and theoretical analysis, but might lead to suboptimal results.

\section{Proof of Privacy (Theorem \ref{thm:privacy_analysis})}
\label{apx:proof_thm_privacy_analysis}

We utilize GDP (introduced in Appendix \ref{apx:GDP}) as our privacy accountant tool. The two cases in Theorem \ref{thm:privacy_analysis} are discussed as follows.

\textbf{Case 1: one server and multiple clients are corrupted.}
We first consider the case when server $\mathsf{S_A}$ is corrupted by  the attacker (the case of corrupting server $\mathsf{S_B}$  is similar).    Recall that server $\mathsf{S_A}$ receives the share $[\Delta\theta_t^i]_\mathsf{A}$ from client $\mathsf{C}_i~(\forall i\in\mathcal{I}_t)$  and the noisy share aggregation $[\sum_{i\in\mathcal{I}_t^*}\Delta\theta_t^i]_\mathsf{B}+[\xi_t^\mathsf{B}]_\mathbb{F}$ from server $\mathsf{S_B}$. To infer one record of the victim client $\mathsf{C}_i~(i\in\mathcal{I}_t^*)$, the attacker can obtain the maximum information from the computation of $\sum\nolimits_{i\in\mathcal{I}_t^*}\Delta\theta_t^i+\xi_t^\mathsf{B}$ by adding $[\sum\nolimits_{i\in\mathcal{I}_t^*}\Delta\theta_t^i]_\mathsf{A}$ with  $[\sum_{i\in\mathcal{I}_t^*}\Delta\theta_t^i]_\mathsf{B}+\xi_t^\mathsf{B}$ and converting the result into a real vector. In the most strong case, where all clients (except the victim $\mathsf{C}_i$) are corrupted, the best that the attacker can do is to compute  $\Delta\theta_t^i+\xi_t^\mathsf{B}$. On the other hand,  since server $\mathsf{S_A}$ has the information of which iterations that client $\mathsf{C}_i$ participates in,  the attacker only needs to consider these iterations to infer one record of client $\mathsf{C}_i$.  By leveraging Lemma \ref{lem:GDP_privacy_account}, we obtain the following Lemma.

\begin{lemma}[Privacy against one server and any number of clients]
\label{lem:privacy_server_clients}
    Assume the attacker corrupts one server (either $\mathsf{S_A}$ or $\mathsf{S_B}$, but not the both) and any number of clients (except the victim client $\mathsf{C}_i$). Then, for the benign client $\mathsf{C}_i$ who subsamples one record with probability $p_i$, Algorithm \ref{alg:framework} asymptotically satisfies record-level $\mu_i$-GDP with privacy parameter $\mu_i=p_i\sqrt{T_i(e^{1/\sigma^2}-1)}$, where $T_i$ is the total number of iterations that client $\mathsf{C}_i$ participates in (i.e., selected by the server).  
\end{lemma}
\begin{proof}
    Assume the attacker corrupts $\mathsf{S_A}$. Denote the set of indices of corrupted clients as $\mathcal{A}$, where $i\notin\mathcal{A}$ (i.e., except the victim $\mathsf{C}_i$). For any iteration $t$ that involves client $\mathsf{C}_i$ (i.e., $i\in\mathcal{I}_t^*$), the attacker is able to compute $\sum\nolimits_{i\in\mathcal{I}_t^*,i\notin\mathcal{A}}\Delta\theta_t^i+\xi_t^\mathsf{B}$, where $\xi_t^\mathsf{B}\sim\mathcal{N}(0,(R\sigma)^2\cdot \mathbf{I})$. Since adding or removing one record in client $\mathsf{C}_i$'s dataset $D_i$ would change the value of $\Delta\theta_t^i$  by at most $R$ in terms of $\ell_2$-norm due to the record-level clipping on $g_s$ (in Algorithm \ref{alg:local_updates}), thus $\sum\nolimits_{i\in\mathcal{I}_t^*,i\notin\mathcal{A}}\Delta\theta_t^i$ will be changed by at most $R$ (i.e., the record-level sensitivity of $\sum\nolimits_{i\in\mathcal{I}_t^*,i\notin\mathcal{A}}\Delta\theta_t^i$ is $R$). We note that the client-level clipping (with bound $C$) might reduce the influence of one record on the final submission $\Delta\theta_t^i$, but never enlarge it. To see why, let's do some calculations. Based on the notations in Algorithm \ref{alg:local_updates}, we denote $\alpha=1/\max\{1,~\|\Delta\tilde{\theta}_t^i\|_2/C\}$. According to Line-6 and Line-7 in Algorithm \ref{alg:local_updates}, we have $\Delta\theta_t^i=-\alpha \sum_{s\in\mathcal{S}}\tilde{g}_s$, where $\|\tilde{g}_s\|_2\leqslant R$ due to the record-level clipping. Therefore, the existence of  one record changes $\Delta\theta_t^i$ at most $\alpha R$ in terms of $\ell_2$-norm, where $\alpha R\leqslant R$ because $\alpha\in(0,1]$ by the definition. Since server $\mathsf{S_A}$ does not know whether a specific record is sampled by client $\mathsf{C}_i$ at these iterations (recall that each record is sampled by client $\mathsf{C}_i$ with probability $p_i$), the privacy amplification of record subsampling holds. Note that client $\mathsf{C}_i$ only participates in partial of all iterations, and we denote $T_i$ as the total number of iterations that client $\mathsf{C}_i$ participates in. According to Lemma \ref{lem:GDP_privacy_account}, with noise $\mathcal{N}(0,(R\sigma)^2\cdot \mathbf{I})$, sensitivity $R$, subsampling probability $p_i$ and total number of iterations $T_i$, the privacy parameter of GDP is $\mu_i=p_i\sqrt{T_i(e^{1/\sigma^2}-1)}$. 
    
    When $\mathsf{S_B}$ is corrupted by the attacker, the result is the same because the two noises $\xi_t^\mathsf{A}$ and $\xi_t^\mathsf{B}$ follow the same distribution.
\end{proof}

\textbf{Case 2: multiple clients are corrupted.}
Compared with servers, one or multiple corrupted clients (except the victim client $\mathsf{C}_i$) possess less information about the private data of client $\mathsf{C}_i$, thus the privacy guarantee against the attacker is stronger. For each client who participates in the $t$-th iteration, he/she is able to observe the current model $\theta_t$. After observing both $\theta_t$ and $\theta_{t+1}$, one corrupted client $\mathsf{C}_j~(j\neq i)$ can compute  $\sum\nolimits_{i\in\mathcal{I}_t^*, i\neq j}\Delta\theta_t^i+\xi_t^\mathsf{A}+\xi_t^\mathsf{B}$, where $\xi_t^\mathsf{A}+\xi_t^\mathsf{B}\sim\mathcal{N}(0,2\cdot(R\sigma)^2 \mathbf{I})$ (note that  client $\mathsf{C}_j$ knows $\Delta\theta_t^j$ but does not know $\xi_t^\mathsf{A}$ or $\xi_t^\mathsf{B}$). The following theorem shows that if any number of clients (excludes client $\mathsf{C}_i$) try to collaboratively infer one record of client $\mathsf{C}_i$, they cannot do better than keeping track of all $\theta_t~(t=1,\cdots,T)$, which is because the Gaussian noise is added by the servers and other clients do not know whether a record is sampled by client $\mathsf{C}_i$.

\begin{lemma}[Privacy against any number of clients]
\label{lem:privacy_clients}
	Assume the attacker corrupts any number of clients (except the victim client $\mathsf{C}_i$). Then, for a benign client $\mathsf{C}_i$ who subsamples one record with probability $p_i$, Algorithm \ref{alg:framework} asymptotically satisfies record-level $\mu_i$-GDP with privacy parameter $\mu_i=qp_i\sqrt{T(e^{1/(2\sigma^2)}-1)}$, where $q$ is the probability that a client is selected by the servers in each iteration.
\end{lemma}
\begin{proof}
    Similar to the reasoning in the proof of Lemma \ref{lem:privacy_server_clients}, adding or removing one record in client $\mathsf{C}_i$'s dataset $D_i$ would change the value of $\Delta\theta_t^i$  by at most $R$ in terms of $\ell_2$-norm. For any group of malicious clients $\mathcal{A}$ (except the victim client $\mathsf{C}_i$), they are able to compute $\sum\nolimits_{i\in\mathcal{I}_t^*,i\neq\mathcal{A}}\Delta\theta_t^i+\xi_t^\mathsf{A}+\xi_t^\mathsf{B}$, where $\xi_t^\mathsf{A}+\xi_t^\mathsf{B}\sim\mathcal{N}(0,2\cdot(R\sigma)^2 \mathbf{I})$. Thus, no matter what the size of $\mathcal{A}$, the sensitivity  is $R$ and the added Gaussian noise is $\mathcal{N}(0,2\cdot(R\sigma)^2 \mathbf{I})$. On the other hand,  corrupted clients in $\mathcal{A}$ do not know whether client $\mathsf{C}_i$ participate at $t$-th iteration (recall that each client is sampled by the servers with probability $q$), thus in their perspective, one specific record of client $\mathsf{C}_i$ is subsampled with probability $q\cdot p_i$. According to Lemma \ref{lem:GDP_privacy_account}, with noise $\mathcal{N}(0,2\cdot(R\sigma)^2 \mathbf{I})$, sensitivity $R$, subsampling probability $q\cdot p_i$ and total number of iterations $T$, the privacy parameter of GDP is $\mu_i=qp_i\sqrt{T(e^{1/(2\sigma^2)}-1)}$.
\end{proof}

\textbf{Putting two cases together.}
The value of $\mu_i$ in Lemma \ref{lem:privacy_server_clients} and Lemma \ref{lem:privacy_clients} are different for different assumption of the attacker, and by putting them together we obtain the result in \eqref{equ:mu_i}. Finally, by converting $\mu$-GDP into $(\epsilon,\delta)$-DP (via Lemma \ref{lem:GDP_to_DP}), we finished the proof of Theorem \ref{thm:privacy_analysis}.

\section{Proof of Robustness (Theorem \ref{thm:robustness_analysis})}
\label{apx:proof_thm_robustness_analysis}

In this section, we show the robustness of PRECAD against model poisoning attacks by leveraging the property of \emph{client-level} DP, though client-level privacy is not our objective (recall that our privacy goal is to provide \emph{record-level} privacy). In the following, we first show that releasing the final model $\theta$ satisfies client-level GDP (in Lemma \ref{lem:client_level_GDP}). Then, we utilize this result to prove the robustness bounds (i.e., Theorem \ref{thm:robustness_analysis}), which shows that the attacker neither can increase nor reduce the \emph{expected} loss of a record too much.

\begin{lemma}[Client-level GDP]
	\label{lem:client_level_GDP}
	The final model $\theta=\theta_T$ of Algorithm \ref{alg:framework} asymptotically satisfies client-level $\mu$-GDP for all clients (including malicious clients who implement the model poisoning attack) with privacy parameter
	\begin{align}
	    \mu=q\sqrt{T(e^{1/\tilde{\sigma}^2}-1)}
	\end{align}
	where $\tilde{\sigma}=\sqrt{2}\sigma R/C$.
\end{lemma}
\begin{proof}
    We analyze the client-level privacy by keeping track of the model updates at all iterations. Note that adding or removing one client $\mathsf{C}_j$'s output $\Delta\theta_t^j$ would change the value of $\sum\nolimits_{i\in\mathcal{I}_t^*}\Delta\theta_t^i$  by at most $C$ in terms of $\ell_2$-norm due to the \emph{client-level} clipping on $\Delta\theta_t^i$ (in Algorithm \ref{alg:local_updates}). It even holds for malicious clients who do not clip his/her output $\Delta\theta_t^j$ by $C$, which will be detected and rejected by servers (via secure validation), then it won't be included in the set $\mathcal{I}^*_t$. Thus, the client-level sensitivity of $\sum\nolimits_{i\in\mathcal{I}_t^*}\Delta\theta_t^i$ is $C$. On the other hand, the noise added in $\sum\nolimits_{i\in\mathcal{I}_t^*}\Delta\theta_t^i$ can be rewritten as $\xi_t^\mathsf{A}+\xi_t^\mathsf{B}\sim\mathcal{N}(0,(C\tilde{\sigma})^2 \mathbf{I})$ with $\tilde{\sigma}=\sqrt{2}\sigma R/C$ because $\xi_t^\mathsf{A}+\xi_t^\mathsf{B}\sim\mathcal{N}(0,2\cdot(R\sigma)^2 \mathbf{I})$. Moreover, each client $\mathsf{C}_j$ is subsampled with probability $q$. According to Lemma \ref{lem:GDP_privacy_account}, with noise $\mathcal{N}(0,(C\tilde{\sigma})^2 \mathbf{I})$, client-level sensitivity $C$, client-level subsampling probability $q$ and total number of iterations $T$, the privacy parameter of client-level GDP  is $\mu=q\sqrt{T(e^{1/\tilde{\sigma}^2}-1)}$.
\end{proof}

Now, we are ready to prove the robustness bounds against model poisoning attacks, where the attacker has full control over $K$ corrupted (malicious) clients.
\begin{proof}[Proof of Theorem \ref{thm:robustness_analysis}]
	According to Lemma \ref{lem:client_level_GDP}, randomized mechanism $\mathcal{M}$ satisfies  client-level $\mu$-GDP with privacy parameter $\mu=q\sqrt{T(e^{1/\tilde{\sigma}^2}-1)}$, where $\tilde{\sigma}=\sqrt{2}\sigma R/C$. Then, by group privacy of GDP (in Lemma \ref{lem:GDP_group}), the mechanism $\mathcal{M}$ satisfies $K\mu$-GDP for groups of size $K$. Applying Lemma \ref{lem:GDP_to_DP},  $\mathcal{M}$ satisfies $(\epsilon,\delta(\epsilon))$-DP (with client-level group size $K$) for all $\epsilon\geqslant0$, where
	\begin{align*}
		\delta(\epsilon)=\Phi\left(-\frac{\epsilon}{K\mu}+\frac{K\mu}{2}\right)-e^\epsilon\cdot\Phi\left(-\frac{\epsilon}{K\mu}-\frac{K\mu}{2}\right)
	\end{align*}
	Then, for all $\epsilon\geqslant0$, we have
	\begin{align*}
		&\mathbb{E}_{\theta\sim\mathcal{M}(\mathcal{C}^\prime_K)}[\ell(\theta,z)] \\
		=~ &\int_{0}^{B}\mathbb{P}_{\theta\sim\mathcal{M}(\mathcal{C}^\prime_K)}[\ell(\theta,z)>t] \mathrm{d}t\\
		\leqslant~ & e^\epsilon\int_{0}^{B}\mathbb{P}_{\theta\sim\mathcal{M}(\mathcal{C})}[\ell(\theta,z)>t] \mathrm{d}t + \int_{0}^{B} \delta(\epsilon)\mathrm{d}t\\
		=~ &e^\epsilon\mathbb{E}_{\theta\sim\mathcal{M}(\mathcal{C})}[\ell(\theta,z)] + B\delta(\epsilon)
	\end{align*}
	By denoting $\mathcal{L}$ and $\mathcal{L}^\prime_K$ as in \eqref{equ:expected_loss}, we have
	\begin{align*}
		\mathcal{L}^\prime_K\leqslant \inf_{\epsilon\geqslant0}~e^\epsilon\cdot\mathcal{L}+ B\delta(\epsilon) 
	\end{align*}
	which finishes the proof of \eqref{equ:robust_upper_bound}.
	
	Similarly, for all $\epsilon\geqslant0$, we have 
	\begin{align*}
		&\mathbb{E}_{\theta\sim\mathcal{M}(\mathcal{C})}[\ell(\theta,z)]
		\leqslant e^\epsilon\mathbb{E}_{\theta\sim\mathcal{M}(\mathcal{C}^\prime_K)}[\ell(\theta,z)] + B\delta(\epsilon) \\
		\Rightarrow\quad &
		\mathbb{E}_{\theta\sim\mathcal{M}(\mathcal{C}^\prime_K)}[\ell(\theta,z)] \geqslant e^{-\epsilon}(\mathbb{E}_{\theta\sim\mathcal{M}(\mathcal{C})}[\ell(\theta,z)]-B\delta(\epsilon) )
	\end{align*}
	Then,
	\begin{align*}
		\mathcal{L}^\prime_K\geqslant \sup_{\epsilon\geqslant0}~e^{-\epsilon}\cdot(\mathcal{L}- B\delta(\epsilon) )
	\end{align*}
	which finishes the proof of \eqref{equ:robust_lower_bound}.
\end{proof}

\section{Experimental Details}
\label{apx:experimental_details}

\textbf{Backdoor Attacks}. Following \cite{bagdasaryan2020backdoor}, we suppose that the attacker wants the global model to misclassify the images with backdoor features as the targeted labels (i.e., backdoor subtask) while classifying other inputs correctly (i.e., main task). To achieve this, the attacker trains the backdoor model $\theta^*$ on a mix of \emph{benign examples} (i.e., correctly labeled) and \emph{backdoor examples} (i.e., with backdoor features and modified labels), and then inject it to the global model via a model-replacement methodology (refer to Sec. \ref{sec:threat_model}). The benign examples can be obtained via accessing the original training dataset of the corrupted (malicious) clients, and the backdoor examples can be obtained by modifying the benign examples or generating new images with the backdoor feature. In this paper, we consider two types of backdoor attacks, which are described as follows:

1) Pixel-Pattern Backdoor (MNIST dataset) \cite{gu2017badnets}. The backdoor examples are generated by changing four  pixels in the bottom-right corner of all training images (of corrupted clients) from black to white and labeling them as 0, shown in Figure \ref{fig:backdoor_image_mnist}. Since the backdoor feature \emph{white pixels in the bottom-right corner}  does not occur in the original test dataset, we test the accuracy of backdoor subtask on a similarly modified version of the original test dataset (i.e., with white pixels in the bottom-right corner and label 0). Therefore,  pixel-pattern backdoor needs to modify images in both training dataset and test dataset, as well as the labels in the training dataset. 

2) Semantic Backdoor (CIFAR-10 dataset) \cite{bagdasaryan2020backdoor}.  The attacker considers \emph{cars} with vertical stripes on the background wall (shown in Figure \ref{fig:backdoor_image_wall}) as the backdoor feature and labels them as \emph{birds} to generate backdoor examples. Following \cite{bagdasaryan2020backdoor}, we split the data so that only the attacker has training images with backdoor feature, which avoids the global model to forget the backdoor too fast. Considering the original test dataset only contains 4 images with such backdoor feature, we measure the test accuracy of semantic backdoor on 1,000 randomly rotated and cropped version of the 4 backdoor images. Compared with pixel-pattern backdoor, the advantage of semantic backdoor is that the attacker does not need to modify images.

\begin{figure}[!t]
    \centering
    \includegraphics[width=1.8in]{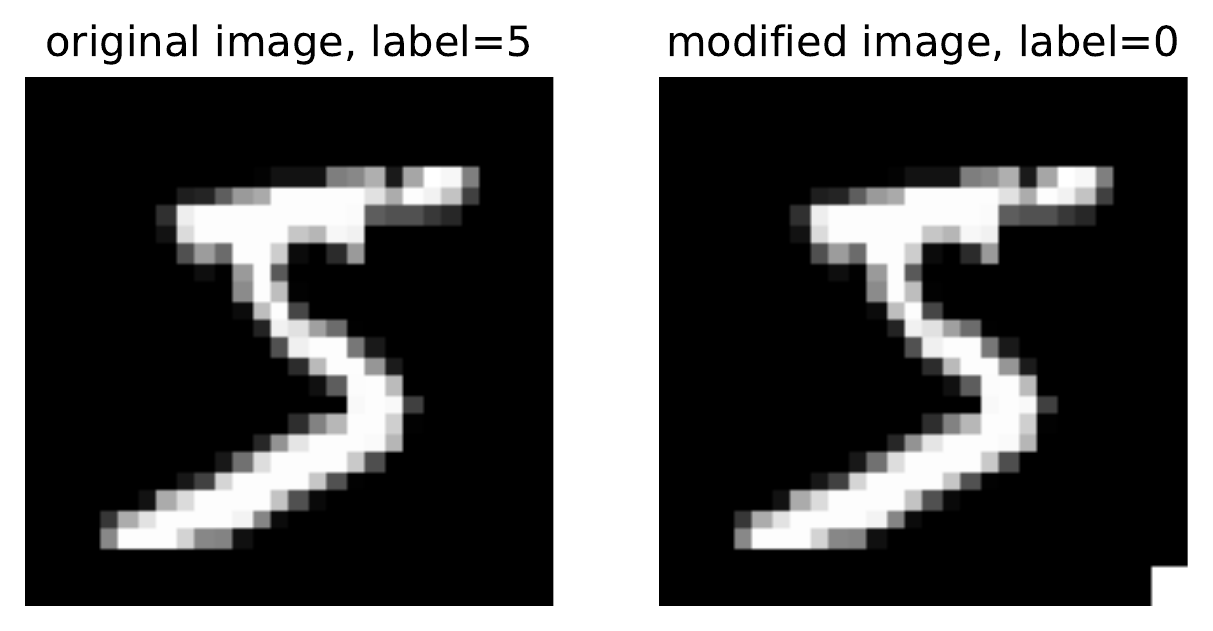}
    \vspace{-2mm}
    \caption{Pixel-pattern backdoor (MNIST dataset) flips   four pixels in the bottom-right corner (with size $2\times2$) from black to white and modifies the label to 0.}
    \vspace{-2mm}
    \label{fig:backdoor_image_mnist}
\end{figure}

\begin{figure}[!t]
    \centering
    \includegraphics[width=2.4in]{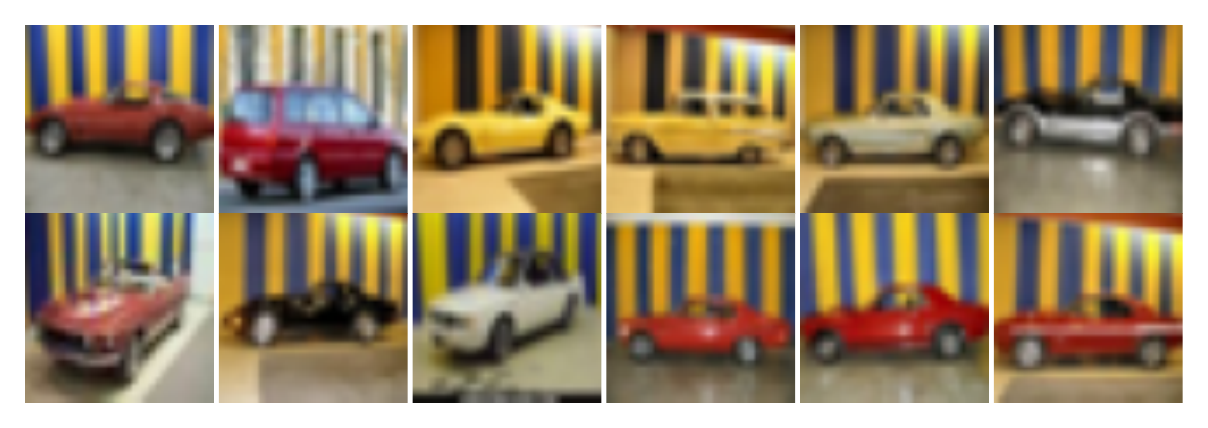}
    \hfill
    \includegraphics[width=0.86in]{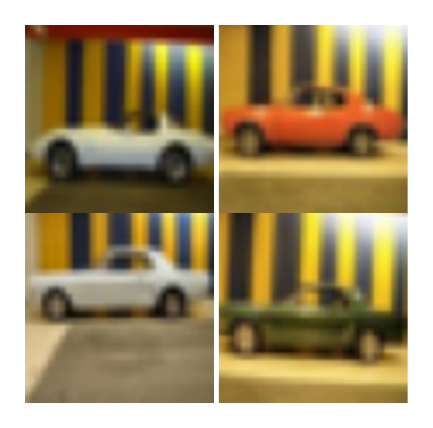}
    \vspace{-2mm}
    \caption{Semantic backdoor (CIFAR-10 dataset) modifies the labels of \emph{cars} with vertical stripes on background wall as \emph{birds}. The training dataset contains 12 backdoor examples (left) and the test dataset contains 4 backdoor examples (right).}
    \label{fig:backdoor_image_wall}
\end{figure}

\textbf{Backdoor Accuracy.}
For MNIST dataset, backdoor accuracy is the fraction of classifying images with \emph{white bottom-right pixel} as 0 (shown in Figure \ref{fig:backdoor_image_mnist}). For CIFAR-10 dataset, backdoor accuracy is the fraction of classifying 1,000 randomly rotated and cropped versions of 4  backdoor images (with true label \emph{cars}, shown in Figure \ref{fig:backdoor_image_wall})  as \emph{birds}.

\textbf{Hyperparameters for Training.}
We fix the leaning rate of the global model as $\eta=0.1$, the probability of each client is selected by the server as $p=0.1$, the probability of each record is sampled by client $\mathsf{C}_i$ as $q_i=0.05$. For private training, we set record-level clipping norm bound $R=2$. Note that the value of $R$ (which is the sensitivity for record-level privacy) does not affect the privacy accountant of record-level privacy because the standard deviation of added Gaussian noise is proportional to $R\cdot\sigma$, thus the privacy accountant is adjusted by the value of noise multiplier $\sigma$. For the training of the backdoor model, we set the learning rate as $0.02$  and the number of \emph{local} iterations as $L=5$ and (vs. $L=1$ for the benign clients). To reduce the influence of randomness on the evaluation, we repeat each experiment with 10 runs and show the averaged results.

\end{document}